%% file: ms.tex
\pdfoutput=1
\documentclass[%
 reprint,           
 superscriptaddress,
 amsmath,amssymb,
 aps,
 prapplied,         
 nofootinbib,       
]{revtex4-2}

\usepackage{graphicx}   
\usepackage{dcolumn}    
\usepackage{bm}         
\usepackage{amsmath}
\usepackage[breaklinks=true]{hyperref}   

\newtheorem{proposition}{Proposition}
\newtheorem{definition}{Definition}
\newtheorem{theorem}{Theorem}
\newtheorem{lemma}{Lemma}
\newtheorem{corollary}{Corollary}

\graphicspath{{figures/}}

\begin{document}

\title{Casimir--electrostatic pull-in in nanoelectromechanical actuators:
Differentiable design sensitivities and the damping-dependent
collapse boundary}

\author{N.~S. Akintsov}
\email{akintsov777@ntu.edu.cn}
\altaffiliation[ORCID: ]{0000-0002-1040-1292}
\affiliation{%
 School of Artificial Intelligence and Computer Science,
 Nantong University, Nantong 226019, China
}%

\author{A.~P. Nevecheria}
\email{artiom.nevecherya@gmail.com}
\altaffiliation[ORCID: ]{0000-0001-6736-4691}
\affiliation{%
 Department of Mathematical and Computer Methods,
 Kuban State University, Krasnodar 350040, Russia
}%

\author{S.~N. Andreev}
\email{andreev@cir-innovations.ru}
\altaffiliation[ORCID: ]{0000-0003-3588-2894}
\affiliation{%
 Joint-Stock Company ``Center for Research and Development'',
 Moscow 101000, Russia
}%

\author{Qing-Hua Qin}
\email{qinghua.qin@smbu.edu.cn}
\altaffiliation[ORCID: ]{0000-0003-0948-784X}
\affiliation{%
 Institute of Advanced Interdisciplinary Technology,
 Shenzhen MSU-BIT University, Shenzhen 518172, China
}%

\date{\today}

\begin{abstract}
Nanoelectromechanical actuators operating at sub-100-nm gaps collapse through a
pull-in instability set by competing electrostatic and Casimir forces. The
quasi-static fold that bounds their safe operating range has been known in
closed form for three decades, together with the Casimir ceiling above which no
static equilibrium survives, which fixes the smallest gap a given stiffness and
area can hold open against the quantum vacuum. That fold does not give the
threshold reached from rest, its dependence on damping, or the design
sensitivities of either. We train a physics-informed neural network in a
rapidity coordinate that maps the movable pull-in pole to infinity, which keeps
the residual bounded across the collapse threshold where fixed-step
Runge--Kutta integration steps into unphysical states. Differentiating the
trained surrogate returns pull-in-voltage sensitivities that match the
closed-form fold to a relative error of $3\times10^{-6}$ and inverts a device
specification to a gap of $97.036\,\mathrm{nm}$ at a target actuation voltage.
Applied to the from-rest boundary, which carries no closed form once the damping
is finite, it supplies the same sensitivities where no analytic root exists. We
prove that this boundary is bracketed by two closed-form curves, that it is
nondecreasing in the damping ratio, that it merges with the fold once the
damping ratio exceeds $2^{-1/4}$, and that the gap closes as
$(\tau_*-\tau)^{2/5}$. Numerically the merger already occurs at $0.396$, and
the growth of the collapse time changes there from logarithmic to inverse
square root. The classical-limit bound on the thermal Lifshitz derating is at
the percent level, and the physical shift at these gaps lies orders of
magnitude below it.
\end{abstract}

\maketitle

\begin{figure*}
 \includegraphics[width=\textwidth]{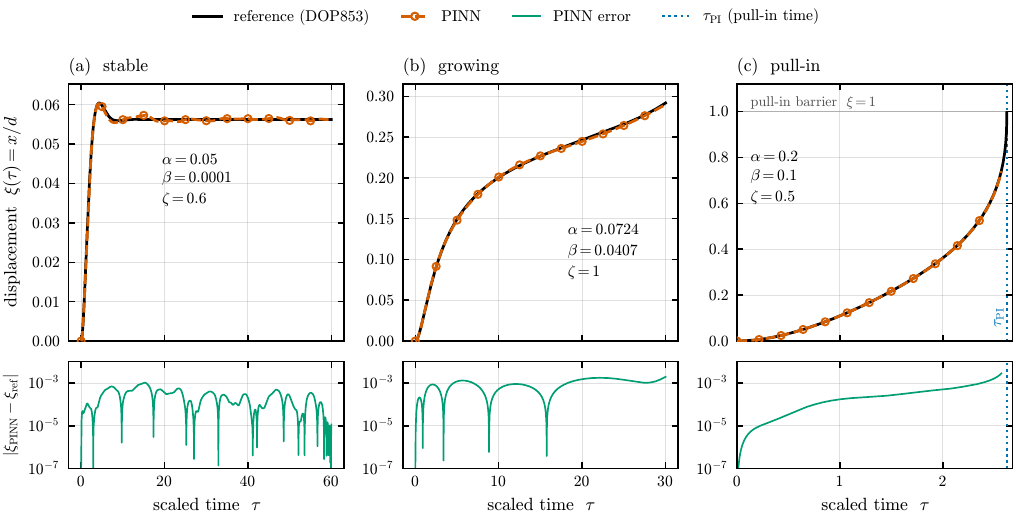}%
 \caption{\label{fig:trajectories}%
  Representative normalized-displacement trajectories $\xi(\tau)=x/d$ for the
  three dynamical regimes of Eq.~\eqref{eq:nd}: (a) stable oscillation
  ($\alpha=0.05$, $\beta=10^{-4}$, $\zeta=0.6$) settling to $\xi\approx0.056$;
  (b) slow monotonic approach to a Casimir-shifted equilibrium
($\xi\approx0.29$)
  ($\alpha=0.0724$, $\beta=0.0407$, $\zeta=1.0$); and
  (c) pull-in collapse ($\alpha=0.20$, $\beta=0.10$, $\zeta=0.5$): the gap
  closes at $\tau_{\mathrm{PI}}=2.6226$ (dotted line), and the last $2\%$ of
  that interval carries $\xi$ from $0.742$ to the barrier $\xi=1$ (gray line).
  Solid curves are the adaptive DOP853 reference, integrated through to the
  barrier; open markers are the physics-informed neural-network (PINN)
  solution, drawn only on its training window $[0,0.98\,\tau_{\mathrm{PI}}]$.
The strip under each panel gives the
  pointwise error $|\xi_{\mathrm{PINN}}-\xi_{\mathrm{ref}}|$ on a common
  five-decade scale, with peak deviations $1.05\times10^{-3}$,
  $1.95\times10^{-3}$, and $2.83\times10^{-3}$.}
\end{figure*}

\input{sections/intro}

\input{sections/model}

\input{sections/pinn}

\input{sections/results}

\begin{figure}
 \includegraphics[width=\columnwidth]{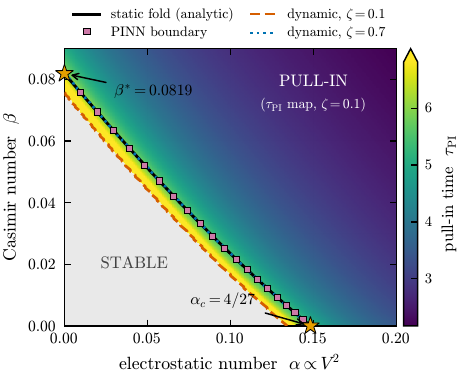}%
 \caption{\label{fig:phase_diagram}%
  Pull-in phase diagram in the $(\alpha,\beta)$ control-parameter plane at
  damping $\zeta=0.1$. The color map gives the dynamic pull-in time
  $\tau_{\mathrm{PI}}(\alpha,\beta)$ for trajectories launched from rest; the
  gray region is stable. The solid line is the analytic quasi-static fold,
  Eq.~\eqref{eq:fold}, the square markers are the boundary located by the
  parametric PINN, and the dashed and dotted lines are the from-rest dynamic
  boundaries at $\zeta=0.1$ and $\zeta=0.7$. Stars mark the endpoints
  $\alpha_c=4/27$ at $\beta=0$ and $\beta^\ast=256/3125=0.0819$ at $\alpha=0$.
  The dynamic boundary at $\zeta=0.1$ lies inside the fold by a
  kinetic-overshoot margin of mean width $0.0117$ in $\alpha$, and
  $\tau_{\mathrm{PI}}$ diverges as the boundary is approached.}
\end{figure}

\begin{figure}
 \includegraphics[width=\columnwidth]{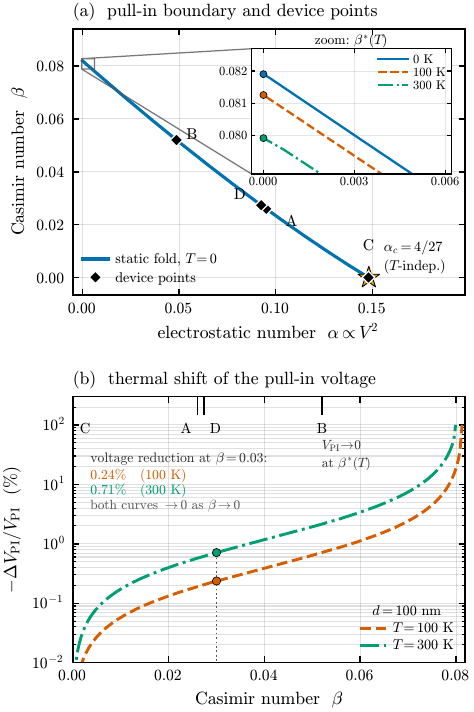}%
 \caption{\label{fig:lifshitz}%
  Classical-limit upper bound on the thermal Lifshitz contraction of the
  pull-in boundary for a $d=100$~nm gap. (a)~Static boundary in the
  $(\alpha,\beta)$ plane at $T=0$ with the four device operating points of
  Table~\ref{tab:devices}. The $\beta=0$ endpoint $\alpha_c=4/27$ is purely
  electrostatic and therefore temperature independent, and device~C
  ($\beta=8.1\times10^{-5}$) sits on it. The inset resolves the Casimir-axis
  intercept, which contracts from $\beta^\ast=0.0819$ at $0$~K to $0.0813$ at
  $100$~K and $0.0799$ at $300$~K. Outside that corner the three boundaries
  separate by less than the line width. (b)~Relative reduction of the
  pull-in voltage,
  $-\Delta V_{\mathrm{PI}}/V_{\mathrm{PI}}=1-\sqrt{\alpha_c(\beta,T)/
  \alpha_c(\beta,0)}$, with both boundaries interpolated onto a common $\beta$
  grid. It vanishes as $\beta\to0$, passes $0.24\%$ ($100$~K) and $0.71\%$
  ($300$~K) at $\beta=0.03$, and grows without bound as $\beta$ approaches
  $\beta^\ast(T)$, where $V_{\mathrm{PI}}\to0$. Ticks along the top axis mark
  the $\beta$ coordinates of devices A--D. Both panels evaluate
  Eq.~\eqref{eq:lifshitz} at a gap far below the thermal wavelength
  ($d\ll\lambda_T=7.6\,\mu$m at $300$~K), so they bound the thermal derating
  rather than report it; the physical correction at this gap follows
  Eq.~\eqref{eq:app-lowT-P} and is five to seven orders of magnitude
  smaller, depending on the temperature.}
\end{figure}

\begin{figure}
 \includegraphics[width=\columnwidth]{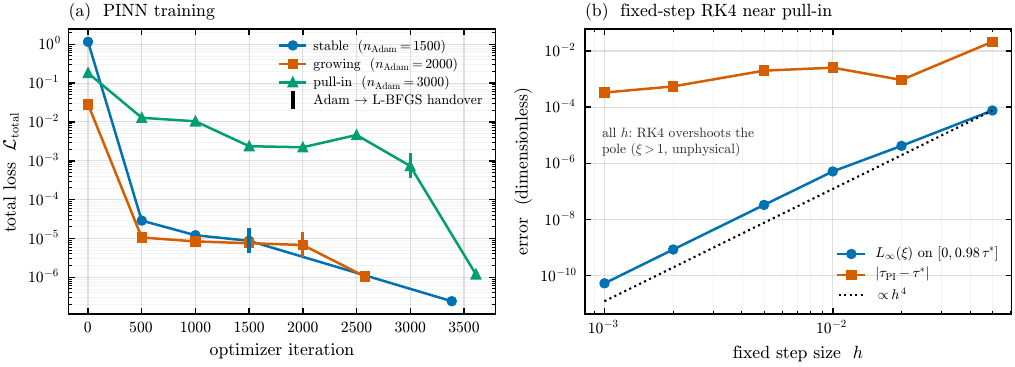}%
 \caption{\label{fig:training}%
  Numerical diagnostics. (a) Composite-loss histories for the three PINN
  trajectory trainings. Each run hands over from Adam to L-BFGS at its own
  iteration, $1500$, $2000$, and $3000$, marked by a vertical tick in the
  color of that curve. The abscissa counts Adam steps up to the tick and
  L-BFGS function evaluations after it, and the L-BFGS phase is logged only
  at its final evaluation, so the closing segment of each curve spans
  unsampled iterations. (b) Fixed-step classical fourth-order Runge--Kutta
  (RK4) near the pull-in pole: for $h\le0.01$ the displacement error follows
  the fourth-order prediction to within $2.5\%$, while every tested step
  $h\in\{0.05,\dots,0.001\}$ overshoots the movable singularity $\xi\ge1$,
  whereas the adaptive reference and the rapidity-coordinate PINN remain well
  posed. The error in $\tau_{\mathrm{PI}}$ is not monotone in $h$ because it
  is set by the detection of the barrier crossing rather than by the order of
  the integrator.}
\end{figure}

\begin{figure*}
 \includegraphics[width=\textwidth]{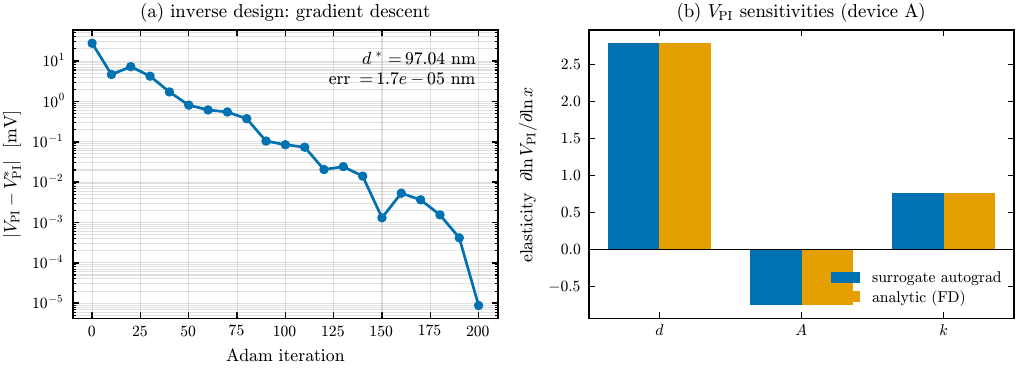}%
 \caption{\label{fig:inverse}%
  The differentiable surrogate as a design tool. (a)~Gradient-based inverse
  design: with area and stiffness fixed, Adam on the gap through the
  differentiable fold network drives the pull-in voltage to the target
  $V_{\mathrm{PI}}^{\ast}=0.30$~V, converging to $d^{\ast}=97.04$~nm and
  reproducing the analytic root to $\sim\!10^{-5}$~nm. (b)~Pull-in-voltage
  elasticities $\partial\ln V_{\mathrm{PI}}/\partial\ln x$ for device~A with
  respect to gap $d$, area $A$, and stiffness $k$, obtained by automatic
  differentiation through the surrogate (blue) and validated against a
  finite-difference evaluation of the analytic pipeline (orange). Each pair
  carries its elasticity and the relative disagreement of the underlying
  derivative, $3.4\times10^{-6}$ or smaller. Table~S5 of the Supplemental
  Material repeats the comparison at device~B, where the Casimir loading is
  twice as large.}
\end{figure*}

\begin{figure}
 \includegraphics[width=\columnwidth]{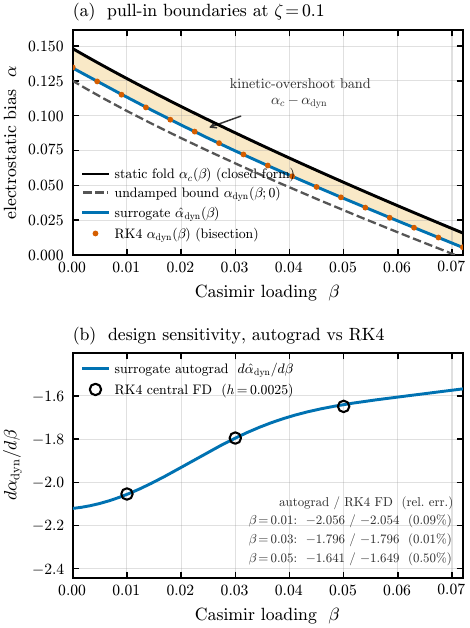}%
 \caption{\label{fig:dynbound}%
  The differentiable surrogate on a boundary that carries no closed form at
  finite damping. (a)~The from-rest dynamic pull-in threshold
  $\alpha_{\mathrm{dyn}}(\beta)$ at $\zeta=0.1$ (vermillion markers, located
  by RK4 bisection, every second point drawn) lies below the closed-form
  static fold $\alpha_c(\beta)$ (black) by the shaded kinetic-overshoot band,
  of mean width $0.0117$ in $\alpha$. The gray dashed curve is the undamped
  from-rest boundary, closed form~\cite{McLellan2016}; it and the fold bracket
  the threshold at every damping ratio (Theorem~\ref{thm:bracket}). A
differentiable surrogate (blue) fitted
  to the $33$ RK4 points reproduces the boundary to a root-mean-square error
  of $1.5\times10^{-4}$ in $\alpha$. (b)~Its automatic-differentiation design
  sensitivity $d\alpha_{\mathrm{dyn}}/d\beta$ (blue), against a central finite
  difference of the independently recomputed boundary at
  $h=2.5\times10^{-3}$ (open circles). The surrogate returns $-2.056$,
  $-1.796$, and $-1.641$ at $\beta=0.01$, $0.03$, and $0.05$, agreeing with
  the finite difference to $0.09\%$, $0.01\%$, and $0.5006\%$. The surrogate returns design sensitivities where the threshold
  is defined only implicitly by the integrated equation of motion.}
\end{figure}

\input{sections/tab_devices}

\input{sections/conclusions}

\begin{acknowledgments}
This work was partially supported by the State Assignment of the Ministry of
Education and Science of the Russian Federation (Project No. FZEN 2023-0006),
the
Nantong Science and Technology Plan Project (Grant Nos. JC2020137 and
JC2020138),
the Key Research and Development Program of Jiangsu Province of China (Grant No.
BE2021013-1), the National Natural Science Foundation of Jiangsu Province of
China
(Grant No. BK20201438), and in part by the Natural Science Research Project of
Jiangsu Provincial Institutions of Higher Education (Grant Nos. 1120KJA510002
and
20KJB510010).

\end{acknowledgments}

Conceptualization, Akintsov; methodology, Akintsov; formal analysis,
Nevecheria, Akintsov and Andreev; investigation, Andreev and Qin; software,
Nevecheria; validation, Nevecheria and Qin; visualization, Akintsov and
Nevecheria; writing of the original draft, Akintsov and Nevecheria; review and
editing, Akintsov and Nevecheria; supervision, Andreev.

The authors declare no competing interests.

\section*{Data availability}
The analysis code, trained network weights, training logs, and the scripts
that regenerate every figure, every table, and every number of the main text
and of the Supplemental Material are openly archived at
Zenodo~\cite{Zenodo2026} and run deterministically from a single fixed seed on
CPU. The archive also carries the checks of Appendix~
ef{sec:app-dyn}
that bear on the barrier constant, on the damping at which the band
closes, and on the departure from the degenerate equilibrium.

\input{sections/appendix_lifshitz}

\input{sections/appendix_dynamics}

\bibliography{refs}

\clearpage

\makeatletter
\let\@hangfrom@section\@hang@from
\let\@sectioncntformat\@seccntformat
\makeatother

\setcounter{section}{0}
\setcounter{equation}{0}
\setcounter{table}{0}
\setcounter{figure}{0}
\renewcommand{\thesection}{S\arabic{section}}
\renewcommand{\thesubsection}{S\arabic{section}\,\Alph{subsection}}
\renewcommand{\theequation}{S\arabic{equation}}
\renewcommand{\thetable}{S\arabic{table}}
\renewcommand{\thefigure}{S\arabic{figure}}

\begin{center}
{\large\bfseries Supplemental Material}
\end{center}

\input{sections/sm_body}

\end{document}

%% file: sections/intro.tex
\section{\label{sec:intro}Introduction}

Electrostatically actuated nanoelectromechanical systems (NEMS) and their
larger microelectromechanical (MEMS) counterparts convert an applied voltage
into controlled mechanical displacement, and they underpin capacitive
switches, resonant sensors, and tunable-gap actuators. In every such device a
movable electrode is driven across a narrow vacuum gap toward a fixed
electrode, and the useful travel range is bounded by an electromechanical
collapse known as pull-in: beyond a critical bias the elastic restoring force
can no longer balance the attractive load, the electrode snaps to contact, and
the device may stick. At gaps below \(100\,\mathrm{nm}\) two attractive loads
compete. The electrostatic pressure of a parallel-plate capacitor scales as
\((d-x)^{-2}\) with the instantaneous gap \(d-x\), while the quantum-vacuum
Casimir force~\cite{Casimir1948} scales as \((d-x)^{-4}\) and therefore
dominates as the gap closes. Both are softening: their gradients reinforce the
displacement rather than oppose it. Precision measurements of the Casimir
interaction in MEMS torsional oscillators~\cite{Chan2001,Decca2007} and the
review of Casimir physics between real
materials~\cite{Klimchitskaya2009} have established that at these separations
the Casimir term is not a perturbation but a leading contribution to the force
balance, a point developed for actuator design in a series of NEMS pull-in
studies~\cite{Batra2007,Batra2008,Koochi2010,Zhang2014,Javor2021,Xu2022,Elsaka2024}.

Pull-in is, mathematically, a saddle-node (fold) bifurcation: as the drive
increases, the stable equilibrium and an unstable saddle approach, merge, and
annihilate at a critical normalized displacement \(\xi_{\mathrm{PI}}=x/d\).
For a purely electrostatic actuator the classical value is
\(\xi_{\mathrm{PI}}=1/3\). Adding the steeper Casimir term moves the fold
to
smaller displacement and smaller bias, lowering the safe operating window. The
governing equation carries a movable singularity at \(\xi\to1\), where the gap
vanishes and the forces diverge. Series expansions, Galerkin reductions, and
other analytic approximations are accurate for small displacement but lose
uniform validity as the trajectory approaches this pole, the regime that sets
the pull-in threshold.

The quasi-static fold itself is settled. Serry, Walliser, and Maclay fixed the
Casimir-only collapse at \(\xi_{\mathrm{PI}}=1/5\) and the ceiling
\(256/3125\)~\cite{Serry1995}; Buks and Roukes gave the first-order shift of
the classical electrostatic threshold \(4/27\)~\cite{Nathanson1967} in the
Casimir strength~\cite{BuksRoukes2001}; and the locus interpolating the two
endpoints appears in Refs.~\cite{Palasantzas2005,Batra2007,LinZhao2007}, the
last of which writes the same damped equation with the same two control
parameters. What the fold does not supply is the behavior away from it: the
collapse time across the control plane, the threshold reached from rest rather
than through a slow voltage ramp, and the design sensitivities of either. All
three require the trajectory, and the trajectory is what the movable pole
obstructs.

Physics-informed neural networks (PINNs) offer a mesh-free route to this
problem. A PINN represents the solution of an ordinary differential equation
(ODE) by a neural network and trains it by minimizing a loss that embeds the
ODE residual together with the initial and boundary
constraints~\cite{Raissi2019,Wang2021,McClenny2023,Arzani2023,CayusoPRD2025},
so that the governing physics is enforced directly rather than through a
discretized time step. This formulation is attractive here because a
fixed-step integrator such as the fourth-order Runge--Kutta method (RK4)
becomes stiff and loses accuracy as the trajectory nears the movable pole,
where the right-hand side grows without bound. A network whose loss is aware of
the singularity can remain accurate in the same regime.

The pull-in barrier has a structural counterpart in relativistic dynamics. As
\(\xi\to1\) the dimensionless gap \(1-\xi\) plays the role of the factor
\(1-v^2/c^2\) for a charged particle accelerated in a laser field, where the
velocity approaches the light cone \(v\to c\)~\cite{Dodin2003}. This
correspondence, which rests on the rapidity description of relativistic
particle motion~\cite{AkintsovNevecheria2023},
motivates the regularization used here. For the relativistic problem we have
trained a physics-informed network on the same kinematics with the Lorentz
invariant imposed as a soft penalty~\cite{AkintsovPPCF2026}; the present work
instead builds the rapidity into the parametrization, so that the movable pole
is removed rather than penalized.

The applied payoff is concrete, and it lives in the design and methodology
layer over pull-in physics that is already well characterized. Locating the
fold across \((\alpha,\beta)\) yields the pull-in voltage as a design curve and
quantifies how the always-on Casimir load limits miniaturization at a fixed
stiffness and area. From this chart we assemble a differentiable design pipeline
for Casimir-limited NEMS actuation: the Casimir ceiling \(\beta^{*}\) sets a
floor on the achievable gap, automatic differentiation through the trained
surrogate returns exact design sensitivities, and gradient descent through the
same object inverts a device specification. Two features carry the contribution.
The surrogate is singularity-aware: parametrized in a rapidity coordinate, it
stays well posed on a pole-free domain in the near-collapse regime where
fixed-step integration steps into unphysical states, an advantage that does not
depend on whether a closed form exists. And because the static fold has a closed
form, we use it as a ground-truth oracle that checks every link of the pipeline,
a strength rather than a limitation. Once finite-temperature Lifshitz
corrections are included, the chart also bounds a temperature dependence of the
stability threshold. The remainder of the paper
is organized as follows: Sec.~\ref{sec:model} sets out the physical model, its
dimensionless form, the thermal correction, and the relativistic analogy;
Sec.~\ref{sec:pinn} describes the PINN methodology and its verification against
RK4; Sec.~\ref{sec:results} presents the trajectories, phase diagram, and
thermal boundary shift; and Sec.~\ref{sec:conclusions} summarizes the findings.

%% file: sections/model.tex
\section{\label{sec:model}Physical Model}

We model the actuator as a single-degree-of-freedom (SDOF) lumped
oscillator~\cite{Batra2008b}: a
movable electrode of effective mass \(m\), linear stiffness \(k\), and
mechanical damping \(\gamma_d\), displaced by \(x\) toward a fixed electrode
across an initial vacuum gap \(d\). The two electrodes are treated as parallel
plates of area \(A\), the standard approximation when the gap is small compared
with the lateral size.

\subsection{\label{sec:model-eom}Equation of motion}

Newton's second law for the movable electrode reads
\begin{equation}
\label{eq:eom}
m\,\ddot{x}+\gamma_d\,\dot{x}+kx
=\frac{\varepsilon_0 A V_0^{2}}{2\,(d-x)^{2}}
+\frac{\pi^{2}\hbar c\,A}{240\,(d-x)^{4}} .
\end{equation}
The left-hand side collects inertia, viscous damping, and the linear elastic
restoring force. On the right, the electrostatic term is the parallel-plate
capacitor attraction under an applied bias \(V_0\), with \(\varepsilon_0\)
the vacuum permittivity, and the Casimir term is the pressure between ideal
conductors, \(\pi^{2}\hbar c/[240\,(d-x)^{4}]\), integrated over the plate
area~\cite{Casimir1948,Milton2004}. Both loads are attractive and diverge as
the gap \(d-x\) closes. The material and geometry dependence of the Casimir
pressure, neglected in the ideal-conductor form used here, is reviewed
in Refs.~\cite{Klimchitskaya2009,Bordag2009,Batra2007}.

Two modeling choices enter the force model. First, the analysis assumes
the parallel-plate (Decca-type) configuration~\cite{Decca2007}. Sphere--plate
experiments~\cite{Chan2001} carry a different force exponent and are treated
through the proximity-force approximation, which would rescale the geometric
prefactor. Second, the ideal-conductor Casimir pressure used here overestimates
the real gold force at a \(100\,\mathrm{nm}\) gap by close to a factor of
two: the finite-conductivity reduction factor at that separation is
\(\eta_F=0.48\)~\cite{LambrechtReynaud2000}, while surface roughness accounts
for about \(1\%\)~\cite{Klimchitskaya2009,Bordag2009}.
Consequently \(\beta\), \(V_{\mathrm{PI}}\), and the thermal shift computed
below
are ideal-conductor upper bounds. Substituting realistic optical data for gold
would reduce \(\beta\) and rescale all dependent predictions accordingly.

\subsection{\label{sec:model-nd}Dimensionless formulation}

Introduce the normalized displacement \(\xi=x/d\in[0,1)\), the dimensionless
time \(\tau=\omega_0 t\) with natural frequency \(\omega_0=\sqrt{k/m}\), and the
quality factor \(Q=m\omega_0/\gamma_d\). Writing an overdot for
differentiation with respect to \(\tau\) and using \(d-x=d(1-\xi)\),
Eq.~\eqref{eq:eom} becomes
\begin{equation}
\label{eq:nd}
\ddot{\xi}+2\zeta\,\dot{\xi}+\xi
=\frac{\alpha}{(1-\xi)^{2}}+\frac{\beta}{(1-\xi)^{4}} .
\end{equation}
The state \(\xi\to1\) is pull-in, at which the gap collapses. The three
dimensionless groups are defined in Table~\ref{tab:params}.
\begin{table}[t]
\caption{\label{tab:params}Dimensionless control parameters of
Eq.~\eqref{eq:nd}.}
\begin{ruledtabular}
\begin{tabular}{lcl}
Parameter & Definition & Role \\ \colrule
\(\alpha\) & \(\varepsilon_0 A V_0^{2}/(2kd^{3})\) & electrostatic strength \\
\(\beta\)  & \(\pi^{2}\hbar c\,A/(240kd^{5})\)     & Casimir strength \\
\(\zeta\)  & \(\gamma_d/(2m\omega_0)=1/(2Q)\)      & damping ratio \\
\end{tabular}
\end{ruledtabular}
\end{table}
Here \(\alpha\) is tunable through the bias, whereas \(\beta\) is fixed by the
device geometry once \(k\), \(A\), and \(d\) are set.

Static equilibria are the fixed points of Eq.~\eqref{eq:nd}, i.e. the roots of
\begin{equation}
\label{eq:equilibria}
g(\xi)\equiv\frac{\alpha}{(1-\xi)^{2}}+\frac{\beta}{(1-\xi)^{4}}-\xi=0 .
\end{equation}
A saddle-node (fold) occurs where a stable node and the neighboring saddle
merge, so that \(g=g'=0\) simultaneously. With the pull-in gap
\(u=1-\xi_{\mathrm{PI}}\), these two conditions read
\begin{equation}
\label{eq:fold-system}
\alpha\,u^{-2}+\beta\,u^{-4}=1-u,\qquad
2\alpha\,u^{-3}+4\beta\,u^{-5}=1,
\end{equation}
which are linear in \((\alpha,\beta)\) with determinant
\(4u^{-7}-2u^{-7}=2u^{-7}\neq0\), so solving gives the closed form
\begin{equation}
\label{eq:fold}
\alpha_c(u)=\tfrac12\,u^{2}\,(4-5u),\qquad
\beta_c(u)=\tfrac12\,u^{4}\,(3u-2).
\end{equation}
Positivity of both forces restricts the fold locus to \(u\in[2/3,4/5]\), that
is \(\xi_{\mathrm{PI}}\in[1/5,1/3]\). The endpoints are
\(\alpha_c=4/27\) at \(\beta=0\) (the classical electrostatic pull-in at
\(\xi_{\mathrm{PI}}=1/3\)~\cite{Nathanson1967}) and
\(\beta^{*}=256/3125\approx0.0819\) at \(\alpha=0\)
(\(\xi_{\mathrm{PI}}=1/5\)), beyond which the Casimir force alone collapses the
gap at zero bias.

Equation~\eqref{eq:fold} is a known result, and we restate its derivation only
because every later section is written in the pull-in gap \(u\). The
\(\alpha=0\) endpoint and the value \(1/5\) are due to Serry, Walliser, and
Maclay~\cite{Serry1995}; the slope at the opposite endpoint is due to Buks and
Roukes~\cite{BuksRoukes2001}; and the locus itself appears in
Refs.~\cite{Palasantzas2005,Batra2007,LinZhao2007}, of which
Ref.~\cite{LinZhao2007} treats the same damped oscillator with the same
\(\alpha\) and \(\beta\). The fold
locus \(\{\alpha_c(u),\beta_c(u)\}\) is the pull-in boundary mapped in
Fig.~\ref{fig:phase_diagram}. As a representative operating point, a gold NEMS
plate with \(A\sim100\,\mu\mathrm{m}^{2}\), \(d\sim100\,\mathrm{nm}\), and
\(k\sim0.5\,\mathrm{N/m}\) gives \(\beta\approx0.026\) and a pull-in voltage
\(V_{\mathrm{PI}}=\sqrt{2\alpha_c kd^{3}/\varepsilon_0
A}\approx0.33\,\mathrm{V}\).

\subsection{\label{sec:model-lifshitz}Thermal Lifshitz correction}

At finite temperature the Casimir pressure acquires a thermal correction from
the Lifshitz theory of fluctuating fields~\cite{Lifshitz1956}. In the classical,
high-temperature limit (the regime in which the gap approaches the thermal
wavelength \(\lambda_T=\hbar c/(k_B T)\), so that \(d\gtrsim\lambda_T\)), the
leading term enhances the Casimir strength by a factor evaluated at the
instantaneous gap,
\begin{equation}
\label{eq:lifshitz}
\beta\;\longrightarrow\;
\beta\!\left[\,1+\frac{60\,\zeta(3)}{\pi^{3}}\,
\frac{k_B T\,(d-x)}{\hbar c}\right]
\qquad(d\gtrsim\lambda_T),
\end{equation}
with \(\zeta(3)\approx1.202\) the Ap\'ery constant\footnote{Here \(\zeta(3)\)
denotes the Riemann zeta value (Ap\'ery's constant), not the damping ratio
\(\zeta\).} and
\(c_L\equiv60\zeta(3)/\pi^{3}\approx2.33\), derived in
Appendix~\ref{sec:app-lifshitz} as the ratio of the classical Casimir pressure
to its zero-temperature value. Because the bracket is linear in
\((d-x)=d(1-\xi)\), it adds a \((1-\xi)^{-3}\) contribution sitting between the
electrostatic and Casimir terms,
\begin{equation}
\label{eq:lifshitz-split}
\frac{\beta\,[1+\kappa T(1-\xi)]}{(1-\xi)^{4}}
=\frac{\beta}{(1-\xi)^{4}}+\frac{\beta\kappa T}{(1-\xi)^{3}},
\end{equation}
where \(\kappa(d)=c_L\,k_B d/(\hbar c)\); at
\(d=100\,\mathrm{nm}\), \(\kappa=1.016\times10^{-4}\,\mathrm{K}^{-1}\).

The relevant scale is the thermal wavelength \(\lambda_T=\hbar c/(k_B T)=
7.6\,\mu\mathrm{m}\) at \(300\,\mathrm{K}\). At the sub-100-nm gaps that are the
focus of this work, \(d\ll\lambda_T\), so the device operates deep in the
low-temperature regime where Eq.~\eqref{eq:lifshitz} does not apply: the true
thermal correction to the ideal-metal Casimir force is higher-order in
\(d/\lambda_T\) and is sub-percent, as established both experimentally and
theoretically~\cite{Klimchitskaya2009,Bordag2009,Bimonte2014}. We therefore
treat the enhancement of Eq.~\eqref{eq:lifshitz}, evaluated at \(d=100\,
\mathrm{nm}\) (about \(0.76\%\) at \(100\,\mathrm{K}\) and \(2.29\%\) at
\(300\,\mathrm{K}\) at the pull-in gap), as a classical-limit upper bound, not
as a physical prediction at that gap. Where nonzero, the correction is
one-sided and always
destabilizing~\cite{Bezerra2016,Bimonte2014,Klimchitskaya2022}. The
classical-limit upper bound is mapped in Fig.~\ref{fig:lifshitz}.

\subsection{\label{sec:model-analogy}Analogy with relativistic particle
dynamics}

Equation~\eqref{eq:nd} shares its analytic structure with the motion of a
charged particle accelerated in a laser
field~\cite{Dodin2003,AkintsovPPCF2026}.
Both are second-order ODEs whose right-hand side carries a movable pole at an
unreachable boundary---\(\xi\to1\) for the actuator, \(v\to c\) for the
particle---and both admit the same remedy, a rapidity coordinate in which
the boundary recedes to infinity. Here that coordinate is
\(\xi=1-e^{-\theta}\), which maps the pole to \(\theta\to\infty\). The
correspondence is summarized in Table~\ref{tab:analogy}.
\begin{table}[t]
\caption{\label{tab:analogy}Structural correspondence between the NEMS pull-in
problem and the relativistic particle in a laser field.}
\begin{ruledtabular}
\begin{tabular}{ll}
NEMS oscillator & Laser-driven particle \\ \colrule
electrostatic strength \(\alpha\) & normalized amplitude \(a_0\) \\
Casimir strength \(\beta\) & radiation-reaction parameter \(\eta\) \\
pull-in \(\xi\to1\) & light cone \(v\to c\) (movable pole) \\
\((\alpha,\beta)\) fold & \((a_0,\eta)\) energy bifurcation \\
pull-in loss term & light-cone loss regularization \\
\end{tabular}
\end{ruledtabular}
\end{table}
This is a formal, structural analogy between two movable-pole
problems that share a rapidity variable, not a Lorentz-covariant identity: the
NEMS dynamics is nonrelativistic, and the mapping is exploited only to import
the rapidity regularization, familiar from the relativistic description of
particle motion, into the pull-in setting.

%% file: sections/pinn.tex
\section{\label{sec:pinn}Physics-Informed Neural-Network Methodology}

We approximate the single-trajectory solution $\xi(\tau)$ of the pull-in
oscillator with a physics-informed neural network (PINN), a parametric
surrogate whose parameters we fit by minimizing the governing ordinary
differential equation (ODE) residual, evaluated through the network's own
automatic derivatives rather than through labeled solution
data~\cite{Raissi2019}. The surrogate is smooth, mesh-free, and analytically
differentiable. Trained in the rapidity coordinate of
Sec.~\ref{sec:pinn-rapidity}, it represents the approach to pull-in on a
pole-free domain.

\begin{definition}[Collocation residual; movable pole]\label{def:residual}
\sloppy
The \emph{collocation residual} is the pointwise value of the governing ODE
operator applied to the network output at a sampled collocation point. The
mean square of this residual over the sampled points is the physics loss. A
\emph{movable pole} is a singularity of the solution whose location depends
on the initial data and parameters, as opposed to a fixed singularity of the
ODE coefficients.
\end{definition}

\subsection{\label{sec:pinn-arch}Architecture}

The map $\tau\mapsto\xi(\tau)$ is a fully connected multilayer perceptron (MLP)
with four hidden layers of $64$ units each, $\tanh$ activation, and Xavier
initialization. The scalar time input is affinely rescaled to $[-1,1]$ over the
training window and, when needed, augmented with Fourier features
$\{\sin(k\omega\tau),\cos(k\omega\tau)\}$, with $\omega$ the dimensionless
natural frequency. We train in double precision. The
networks are small enough that this carries negligible cost while easing the
$\sim\!10^{-3}$ accuracy target.

The initial conditions are imposed by construction, through a hard ansatz.
Writing
$g(\tau)=1-e^{-\tau}$, so that $g(0)=0$ and $\dot g(0)=1$, we set
\begin{equation}
\xi_{\rm NN}(\tau)=\xi_0+v_0\,g(\tau)+g(\tau)^2\,\mathcal N(\tau;\Theta),
\label{eq:hardic}
\end{equation}
with $\mathcal N$ the raw network of weights $\Theta$. Then
$\xi_{\rm NN}(0)=\xi_0$ and $\dot\xi_{\rm NN}(0)=v_0$ hold by construction, so
the initial-condition loss can be dropped. We optimize with Adam for warm-up and
refine with L-BFGS. All trainings are CPU-only and deterministic from a
fixed seed ($42$). The complete per-regime protocol, including collocation
counts, optimizer schedules, and wall-clock times, is given in the Supplemental
Material~\cite{SupplMat}, Sec.~S1.

Plain $\tanh$ networks proved adequate for short and moderate windows, but the
long-horizon stable regime ($T=60$) floored near
$E_\infty\!\approx\!6.5\times10^{-3}$: a depth-limited $\tanh$ MLP
spends its capacity resolving many oscillation periods. Adding two Fourier input
modes tuned to the damped natural frequency
$\omega_d=\sqrt{1-\zeta^2}\approx0.8$ (not the undamped $\omega=1$), together
with a longer
L-BFGS polish, restored the target accuracy $\sim\!10^{-3}$. The input encoding,
not raw
depth, sets the attainable error on long trajectories.

\subsection{\label{sec:pinn-loss}Composite loss}

We minimize the composite objective
\begin{equation}
\mathcal L=\mathcal L_{\rm ODE}
          +\lambda_1\,\mathcal L_{\rm IC}
          +\lambda_2\,\mathcal L_{\rm pull\text{-}in},
\label{eq:loss}
\end{equation}
with the physics term the mean square of the collocation residual of
Definition~\ref{def:residual} over $N$ collocation points,
\begin{equation}
\begin{aligned}
\mathcal L_{\rm ODE}&=\frac1N\sum_{i=1}^{N}r(\tau_i)^2,\\
r&=\ddot\xi+2\zeta\dot\xi+\xi-\frac{\alpha}{(1-\xi)^2}-\frac{\beta}{(1-\xi)^4},
\end{aligned}
\label{eq:residual}
\end{equation}
where $\dot\xi$ and $\ddot\xi$ come from \texttt{torch.autograd.grad} applied to
the network output. The term $\mathcal L_{\rm IC}$ penalizes departures from
$\xi(0)=\xi_0$, $\dot\xi(0)=v_0$, and the soft barrier
$\mathcal L_{\rm pull\text{-}in}$ penalizes $\xi\ge1-\delta$ (and $\xi<0$).
Under the rapidity coordinate together with the hard-IC parametrization
\eqref{eq:hardic}, both vanish identically and their weights
$\lambda_1,\lambda_2$ are inert; we retain $\mathcal L_{\rm pull\text{-}in}$
only for the raw parametrization, which does not enforce $\xi<1$ by
construction. In the runs reported here we use $N=1600$--$3000$ collocation
points per regime.

A well-documented PINN failure mode is imbalance among the gradients of the loss
terms, which lets one term dominate the parameter update and stalls
training~\cite{Wang2021,McClenny2023}. Removing $\mathcal L_{\rm IC}$ through
the
hard ansatz eliminates the worst of this imbalance at its source. For the soft
terms that remain, we recommend gradient-balanced or self-adaptive weighting
schemes~\cite{Wang2021,McClenny2023} over hand-tuned constants.

\subsection{\label{sec:pinn-rapidity}Rapidity regularization of the movable
pole}

The force terms $(1-\xi)^{-2}$ and $(1-\xi)^{-4}$ carry a movable pole: released
above the fold, the trajectory closes the gap $1-\xi$ at a finite time $\tau_*$
that depends on the initial data and parameters. We remove this coordinate pole
with the substitution
\begin{equation}
\xi=1-e^{-\theta},\qquad \theta=-\log(1-\xi)\ge0,
\label{eq:rapidity}
\end{equation}
which maps $\xi\in[0,1)$ onto $\theta\in[0,\infty)$ so that the barrier
$\xi\to1$ becomes $\theta\to\infty$. Inserting \eqref{eq:rapidity} into the
governing equation, using $(1-\xi)^{-2n}=e^{2n\theta}$ (so the electrostatic and
Casimir terms become $\alpha\,e^{2\theta}$ and $\beta\,e^{4\theta}$), and
multiplying through by $e^{\theta}$ (so that
$e^{\theta}(1-e^{-\theta})=e^{\theta}-1$
and the exponents advance to $3$ and $5$) yields the transformed
ODE
\begin{equation}
\ddot\theta=\dot\theta^{2}-2\zeta\dot\theta-\bigl(e^{\theta}-1\bigr)
           +\alpha\,e^{3\theta}+\beta\,e^{5\theta}
\label{eq:thetaode}
\end{equation}
with right-hand side $f(\theta,\dot\theta)$. The gap $1-\xi=e^{-\theta}$ decays
like the relativistic Lorentz factor
$1/\gamma=\operatorname{sech}\varphi$, i.e.\ as $2e^{-\varphi}$ at large
$\varphi$, so $\theta$ plays the
role of a ``rapidity'' and $\xi=1$ that of a horizon approached but never
reached~\cite{Dodin2003}. This is a formal chart, not
a Lorentz boost: $\theta$ carries no group-theoretic meaning and generates no
boost additivity. The analogy is structural: an unbounded coordinate that
saturates exponentially.

Training in the rapidity coordinate uses the residual of
Eq.~\eqref{eq:thetaode} directly. Writing $r_\theta$ for the $\theta$-domain
collocation residual,
\begin{equation}
\begin{aligned}
r_\theta&=\ddot\theta
  -\bigl[\dot\theta^{2}-2\zeta\dot\theta-(e^{\theta}-1)
         +\alpha\,e^{3\theta}+\beta\,e^{5\theta}\bigr]\\
 &=\ddot\theta+2\zeta\dot\theta+(e^{\theta}-1)-\dot\theta^{2}
   -\alpha\,e^{3\theta}-\beta\,e^{5\theta},
\end{aligned}
\label{eq:residual-theta}
\end{equation}
the corresponding physics loss is the mean square over $N$ collocation points,
\begin{equation}
\mathcal L_{\rm ODE}^{\theta}=\frac1N\sum_{i=1}^{N}r_\theta(\tau_i)^2 .
\label{eq:loss-theta}
\end{equation}
Initial conditions are again imposed exactly by a hard ansatz in the rapidity
variable,
\begin{equation}
\theta_{\rm NN}(\tau)=\theta_0+\tau\,\dot\theta_0+\tau^{2}\,\mathcal
N(\tau;\Theta),
\label{eq:hardic-theta}
\end{equation}
so that $\theta_{\rm NN}(0)=\theta_0$ and $\dot\theta_{\rm NN}(0)=\dot\theta_0$
hold by construction. The stable and growing regimes are trained with the raw
residual~\eqref{eq:residual}. The pull-in regime is trained with
$\mathcal L_{\rm ODE}^{\theta}$ in the rapidity coordinate, to which
Proposition~\ref{prop:reg} applies.

\begin{proposition}[Pole regularization]\label{prop:reg}
Under the map \eqref{eq:rapidity}, the chart $\theta=-\log(1-\xi)$ is a
$C^\infty$ diffeomorphism of $[0,1)$ onto $[0,\infty)$, with smooth inverse
$\xi=1-e^{-\theta}$, and $\theta(\tau)$ obeys
\eqref{eq:thetaode}, whose right-hand side $f$ is real-analytic (entire) on
$\mathbb R^2$ with no singularity at any finite $\theta$. Consequently, on every
compact interval $\theta\in[0,\theta_{\max}]$ the ODE residual, and hence the
physics loss $\mathcal L_{\rm ODE}$, are bounded and $C^\infty$, so the PINN
optimization is posed on a pole-free domain.
\end{proposition}

\noindent\emph{Proof sketch.} Since $1-\xi=e^{-\theta}>0$, the inverse
$\theta=-\log(1-\xi)$ is smooth and strictly increasing
($d\xi/d\theta=e^{-\theta}>0$), giving the diffeomorphism. The right-hand side
of
\eqref{eq:thetaode} is a finite sum of exponentials $e^{k\theta}$ and a
quadratic
in $\dot\theta$; each summand is entire and finite sums are entire, so $f$ has
no
finite-$\theta$ pole, unlike the original field at $\xi=1$. Since $\theta_{\rm
NN}$
and $\dot\theta_{\rm NN}$ are continuous on the compact interval $[0,T]$, they
take values in a compact set $K$, on which the entire right-hand side $f$ is
bounded; hence $\mathcal L_{\rm ODE}^{\theta}$ is finite and smooth. On this
compact $(\theta,\dot\theta)$ set $f$ attains a finite maximum, and the network
with its automatic derivatives is $C^\infty$ in $\tau$, so $r_\theta$ and its
mean square are bounded and smooth.\hfill$\square$

PINNs for ODEs with singular points were applied by Cayuso \textit{et al.}
to the Legendre, hypergeometric, aether, and spherical-accretion
equations~\cite{CayusoPRD2025}, all of which carry fixed regular singular
points whose location is set by the coefficients. Boundary data there are
imposed by a hard ansatz, and for the Legendre case the envelope is
$1-e^{-(x+1)}$, the same construction we use in Eq.~\eqref{eq:hardic}; their
exponential output maps enforce positivity of the solution, and the residual is
minimized in the original variable throughout. The rapidity substitution of
Eq.~\eqref{eq:rapidity} acts on the dependent variable instead. It sends the
movable pole at $\xi=1$, whose location depends on the initial data and on
$(\alpha,\beta,\zeta)$, to $\theta\to\infty$, and the network is trained on the
residual of the transformed equation, Eq.~\eqref{eq:residual-theta}, whose
right-hand side is entire on $\mathbb{R}^{2}$ by
Proposition~\ref{prop:reg}. Related transformations of the independent variable
desingularize self-similar blow-up~\cite{WangBlowup2023,BuddWilliams2010}.
\emph{Remark (scope).}
Proposition~\ref{prop:reg} regularizes the coordinate pole, not the
finite-time blow-up in $\tau$: one still has $\theta(\tau)\to\infty$ as
$\tau\to\tau_*$, since no change of the dependent variable removes a genuine
finite-time escape of the flow. We therefore train on a window $T<\tau_*$
(equivalently up to a target $\xi_{\max}$), where
$\theta\le\theta_{\max}<\infty$
and the hypotheses hold.

\subsection{\label{sec:pinn-verify}Verification against Runge--Kutta}

We take as ground truth an adaptive DOP853 (Dormand--Prince) integrator with
$\text{rtol}=\text{atol}=10^{-12}$. In the pull-in case
($\alpha=0.2$, $\beta=0.1$, $\zeta=0.5$) it halts cleanly at
$\xi\approx0.999996$ (minimum gap $3.9\times10^{-6}$, $\tau_*\approx2.6226$)
without overshoot or non-finite values: an error-controlled solver refines its
step where the field steepens and does not blow up.

Against this reference the PINN attains supremum-norm ($E_\infty$) displacement
errors of
$1.05\times10^{-3}$ (stable), $1.95\times10^{-3}$ (growing), and
$2.83\times10^{-3}$ (pull-in, rapidity coordinate). Here $E_\infty$ and $E_2$
denote the sup-norm and root-mean-square errors against the reference, reserving
$\mathcal L$ for the training loss. Splitting the pull-in
trajectory, the regular part ($\xi<0.5$) reaches $E_\infty\approx
8.5\times10^{-4}$ while the near-pole part ($\xi\ge0.5$) reaches $2.8\times
10^{-3}$. The velocity is the weak spot, and the rapidity chart is why: with
$1-\xi\propto(\tau_*-\tau)^{2/5}$ the chart velocity $\dot\theta$ reaches $7.4$
on this window (Corollary~\ref{cor:rapidity}), so a fixed relative error in
$\theta$ is amplified in $\dot\xi$. Its error grows to $\approx3.4\times
10^{-2}$ as collapse sharpens. Table~\ref{tab:pinn-verify} collects the
per-regime metrics, and Figs.~\ref{fig:trajectories} and~\ref{fig:training}
show the trajectories and the training diagnostics. Section~S4 of the
Supplemental Material~\cite{SupplMat} tabulates the step study behind these
statements.

Fixed-step classical fourth-order Runge--Kutta (RK4) is fourth-order accurate
away from the pole, but at every tested step $h\in\{0.05,\dots,0.001\}$ it
overshoots into the unphysical region $\xi\ge1$ near the movable pole, where
$(1-\xi)^{-4}$ changes sign and produces a spurious ``blow-up.'' The defensible
conclusion is thus representational rather than a blanket accuracy claim:
fixed-step RK4 steps across the movable singularity, whereas the adaptive
reference and the rapidity-regularized PINN both remain well posed.
Training a single-trajectory network costs $\approx1$--$4$~min on CPU, far
exceeding a single adaptive solve. The surrogate is advantageous when its
differentiability and mesh-free continuity are reused across many downstream
sensitivity, inverse, or control evaluations.

\begin{table}[t]
\caption{\label{tab:pinn-verify}%
PINN accuracy against the adaptive DOP853 reference, per regime. The pull-in
trajectory (rapidity coordinate) is split into its regular part ($\xi<0.5$) and
near-pole part ($\xi\ge0.5$). Errors are in the displacement $\xi$ unless marked
$\dot\xi$.}
\begin{ruledtabular}
\begin{tabular}{lccc}
Regime & $E_\infty(\xi)$ & $E_2(\xi)$ & $E_\infty(\dot\xi)$\\
\hline
Stable (raw)              & $1.05\times10^{-3}$ & $3.50\times10^{-4}$ &
$7.8\times10^{-4}$\\
Growing (raw)             & $1.95\times10^{-3}$ & $1.06\times10^{-3}$ &
$1.1\times10^{-3}$\\
Pull-in, regular          & $8.5\times10^{-4}$  & $3.2\times10^{-4}$  & ---\\
Pull-in, near pole        & $2.83\times10^{-3}$ & $1.49\times10^{-3}$ &
$3.4\times10^{-2}$\\
\end{tabular}
\end{ruledtabular}
\end{table}

%% file: sections/results.tex
\section{\label{sec:results}Results and Discussion}

All dimensionless quantities refer to Eq.~\eqref{eq:nd} and the control
parameters $(\alpha,\beta,\zeta)$ of Table~\ref{tab:params}.

\subsection{Time trajectories and the three dynamical regimes}

Figure~\ref{fig:trajectories} shows the displacement $\xi(\tau)$ released from
rest for three operating points that span the qualitative behavior of
Eq.~\eqref{eq:nd}. A stable point ($\alpha=0.05$, $\beta=10^{-4}$,
$\zeta=0.6$) executes a damped approach and settles onto the stable node at
$\xi\approx0.056$, overshooting before the viscous term removes the excess
energy. An intermediate point ($\alpha=0.0724$, $\beta=0.0407$,
$\zeta=1.0$) climbs monotonically to a larger equilibrium near $\xi\approx0.29$
without collapsing: the bias sits below the fold, but the Casimir pre-load
pushes the fixed point well past the small-displacement range. A supercritical
point ($\alpha=0.20$, $\beta=0.10$, $\zeta=0.5$) lies above the fold and
collapses at a finite time. The collapse is abrupt rather than gradual: the
last stretch $\xi:0.74\to1$ occupies
only the final $\sim\!2\%$ of the trajectory. The movable pole fixes the rate,
$1-\xi=(25\beta/6)^{1/5}(\tau_*-\tau)^{2/5}$ to leading order
(Theorem~\ref{thm:exponent}), so the velocity diverges as
$(\tau_*-\tau)^{-3/5}$. Against the adaptive DOP853 reference the
physics-informed neural network (PINN) reproduces all three trajectories to
maximum pointwise errors below $3\times10^{-3}$, the pull-in case in the
rapidity coordinate of Eq.~\eqref{eq:rapidity}; the per-regime metrics are
collected in Table~\ref{tab:pinn-verify}.

\subsection{Pull-in phase diagram}

Figure~\ref{fig:phase_diagram} maps the $(\alpha,\beta)$ plane at fixed damping
$\zeta=0.1$. The color scale gives the dynamic pull-in time
$\tau_{\mathrm{PI}}(\alpha,\beta)$ obtained by integrating Eq.~\eqref{eq:nd}
from rest until the gap closes; the solid curve is the analytic static fold of
Eq.~\eqref{eq:fold}. The grid, the integrator settings, the pull-in criterion,
and the summary statistics quoted below are collected in Sec.~S3 of the
Supplemental Material~\cite{SupplMat}. The dynamic
boundary separates initial conditions that pull in from those that settle,
and it sits inside the static fold by a mean gap $\Delta\alpha\approx
0.0117$ in the electrostatic strength (median $0.0121$, maximum $0.0141$), so a
device released from rest collapses at a bias below the quasi-static threshold.
The excess is kinetic: a mass released above its final equilibrium overshoots
and crosses the barrier that a slow voltage ramp would not reach. Raising the
damping suppresses the overshoot, monotonically in
$\zeta$ (Corollary~\ref{cor:mono}; for $\beta=0$ see also
Ref.~\cite{Flores2017}), and above a finite damping ratio the two boundaries
coincide exactly, because the orbit released from rest then has no speed left at
the fold. That the two never coincide, asserted in
Refs.~\cite{Flores2017,CassaniMiyasita2024}, does not hold; the point is
settled in Appendix~
ef{sec:app-dyn}. Theorem~\ref{thm:barrier} proves this for $\zeta\ge2^{-1/4}$ at every
$\beta$; the threshold itself is lower, $\zeta_c=0.3959$ at $\beta=0$, and the
band closes quadratically in $\zeta_c-\zeta$. The offset measured at
$\zeta=0.7$, $\Delta\alpha=-6.6\times10^{-5}$, is one fifteenth of a grid cell
and of the opposite sign, and is therefore sweep resolution about an exact
zero. Of the sampled plane, a fraction $0.73$
pulls in dynamically against $0.68$ statically, the difference being the
overshoot band ($5.3\%$ of cells). The fold endpoints reproduce the analytic
values $\alpha_c=4/27\approx0.148$ at $\beta=0$ and $\beta^{*}=256/3125\approx
0.082$ at $\alpha=0$ to within one grid cell. The pull-in time is largest
immediately past threshold, up to $\tau_{\mathrm{PI}}\approx15.8$ against a
floor of $2.2$ deep in the collapse region, but the mechanism at this damping is
not the critical slowing down of a saddle-node. At $\zeta=0.1$ the boundary lies
strictly inside the fold, the equilibrium that the orbit passes is a hyperbolic
saddle, and the transit time grows as
$-\lambda_{+}^{-1}\ln(\alpha-\alpha_{\mathrm{dyn}})$ rather than as
$(\alpha-\alpha_c)^{-1/2}$ (Proposition~\ref{prop:slowing}). Fitting the cells
nearest the boundary returns a logarithmic law in every sampled row, with
prefactors $-1.44$ to $-1.26$ against the predicted $-\lambda_{+}^{-1}=-1.362$
at $\beta=0$. The inverse-square-root law is recovered above
$\zeta\approx0.4$, where the two boundaries merge. This slow, force-limited
pull-in
just above threshold is the dynamical counterpart of the softening reported for
electrostatically and Casimir-loaded
microstructures~\cite{Batra2007,Koochi2010,Zhang2014}.

A parametric PINN trained as a pure-residual equilibrium surrogate reproduces
the fold locus of Eq.~\eqref{eq:fold} to a mean deviation
$|\Delta\alpha|\approx8\times10^{-8}$. The fold is a double-root (tangency)
locus, where the pull-in indicator
is second-order insensitive to the field, so boundary agreement flatters the
surrogate. The fair measure of the learned equilibrium field is its
root-mean-square error in the pull-in gap, $\mathrm{RMSE}(u_*)\approx3\times
10^{-3}$ (the pointwise maximum reaches $\approx0.13$, confined to the
degenerate $\alpha,\beta\to0$ corner where the equilibrium minimizer is
ill-conditioned). Both numbers describe the same network, and we quote them
together to avoid overstating the accuracy.

\subsection{Thermal Lifshitz shift: a classical-limit bound}

Figure~\ref{fig:lifshitz} bounds how temperature displaces the fold through the
classical-limit Lifshitz correction of
Eqs.~\eqref{eq:lifshitz}--\eqref{eq:lifshitz-split}. The curves are a
high-temperature upper bound, not the physical shift: they evaluate
Eq.~\eqref{eq:lifshitz} at a $d=100\,\mathrm{nm}$ gap far below the thermal
wavelength $\lambda_T=7.6\,\mu\mathrm{m}$ ($300\,\mathrm{K}$), where the
classical term is not quantitatively valid. Within this bound the boundary
contracts monotonically as $T$ rises, entirely on the Casimir side: the
$\beta=0$ endpoint is temperature independent, while the Casimir-only
intercept moves from $\beta^{*}=0.0819$ ($0\,\mathrm{K}$) to $0.0799$
($300\,\mathrm{K}$). The bound is one-sided and destabilizing, since heating
enhances the effective Casimir attraction through the
$\beta\kappa T/(1-\xi)^{3}$ term of Eq.~\eqref{eq:lifshitz-split}. At a
representative $\beta=0.03$ it caps the
pull-in-voltage reduction (with $V_{\mathrm{PI}}\propto\sqrt{\alpha_c}$) at
$\Delta V_{\mathrm{PI}}/V_{\mathrm{PI}}\approx-0.24\%$ ($100\,\mathrm{K}$) and
$-0.71\%$ ($300\,\mathrm{K}$). This bound is $\beta$-dependent, ranging across
the
individual devices of Table~\ref{tab:devices} from $-0.003\%$ for the
electrostatic-dominated device~C to $-1.07\%$ for the most Casimir-loaded
device~B. For the sub-100-nm devices studied here, however,
$d\ll\lambda_T$ and the true correction to the ideal-metal Casimir force is
higher-order in $d/\lambda_T$ and
sub-percent~\cite{Klimchitskaya2009,Bordag2009,Bimonte2014}. It displaces the
actuation threshold measurably only
as the gap approaches $\lambda_T$ (micron scale) or at elevated
temperature~\cite{Lifshitz1956,Bezerra2016,Bimonte2014,Klimchitskaya2022}.

\subsection{Numerical verification and the movable pole}

Figure~\ref{fig:training} collects the training diagnostics: the composite-loss
histories converge under Adam warm-up and L-BFGS refinement to residuals of
order $10^{-6}$--$10^{-7}$, and the step study contrasts the integrators at the
pole. At every tested step $h\in\{0.05,\dots,0.001\}$ the fixed-step
fourth-order Runge--Kutta (RK4) scheme steps across the movable pole into the
unphysical region $\xi\ge1$, while the adaptive DOP853 reference and the
rapidity-regularized PINN, which keeps $\xi<1$ by construction on the
transformed, pole-free domain of Eq.~\eqref{eq:thetaode}, both stay well posed.
The point is representational rather than a blanket accuracy claim
(Sec.~\ref{sec:pinn-verify}).

\subsection{\label{sec:applied}Implications for device design}

For gold-coated NEMS capacitive actuators and Casimir oscillators the fold
locus supplies predictive pull-in-voltage estimates through
$V_{\mathrm{PI}}=\sqrt{2\alpha_c kd^{3}/
\varepsilon_0 A}$: the representative parameters of
Table~\ref{tab:devices} give $V_{\mathrm{PI}}\approx0.33\,\mathrm{V}$ for the
100-nm and 50-nm gaps and $\approx3.7\,\mathrm{V}$ for the stiffer 200-nm
device, with the Casimir pre-load cutting the threshold below the classical
$\beta=0$ estimate at the smaller gaps. Ignoring it over-predicts the stable
window and invites stiction~\cite{Chan2001,Zhang2014}. Because
$\beta\propto A/(kd^{5})$ and no static equilibrium survives above the Casimir
ceiling $\beta^{*}$, shrinking the gap collapses the available voltage window
as $d^{-5}$ and fixes the smallest gap that a given stiffness and area can hold
open against the quantum vacuum---a constraint that stiffness or area
engineering must respect. The thermal shift of Sec.~\ref{sec:results}\,C
is a classical-limit upper bound, and for the sub-100-nm devices considered
here ($d\ll\lambda_T$) the realistic derating is sub-percent: derating and the
associated on-chip thermometry enter the design budget only as the gap
approaches $\lambda_T$ or at elevated temperature. The critical slowing down
near the fold makes the threshold region a soft, high-responsivity operating
point for near-threshold force, mass, and voltage
sensing~\cite{Javor2021,Xu2022,Elsaka2024}.

The differentiable surrogate is itself a design tool.
Figure~\ref{fig:inverse} shows two uses that a direct integrator does not
provide. First, because
the parametric fold network is smooth in its inputs, automatic differentiation
returns exact design sensitivities at no extra cost. The boundary slope
$d\alpha_c/d\beta$ obtained by differentiating through the network matches the
closed-form value $-1/u^{2}$ to a relative error of $\sim\!10^{-6}$, and the
same
chain rule delivers the pull-in-voltage elasticities
$\partial\ln V_{\mathrm{PI}}/\partial\ln x$ with respect to gap, area, and
stiffness [$+2.79$, $-0.76$, $+0.76$ for device~A of Table~\ref{tab:devices};
Fig.~\ref{fig:inverse}(b)], again validated against the analytic pipeline to
$\sim\!10^{-6}$. At device~B, which carries twice the Casimir loading, the same
elasticities are $+6.04$, $-1.41$, and $+1.41$, so the surrogate tracks a
quantity that more than doubles across the tabulated devices. Second, the surrogate inverts a device specification by
gradient
descent: fixing area and stiffness and targeting
$V_{\mathrm{PI}}=0.30\,\mathrm{V}$,
Adam on the gap through the differentiable network reaches its tolerance in
$201$ steps [Fig.~\ref{fig:inverse}(a)] at $d=97.04\,\mathrm{nm}$, reproducing the
analytic root to femtometer accuracy. Each inverse solve here takes
$1.4\,\mathrm{s}$,
so the advantage over a one-dimensional root find is not raw latency; it is that
the same trained, differentiable object supplies boundary evaluations, exact
sensitivities, and backpropagated inverse designs within one framework. On this
problem the fold has a closed form, and we use it as a ground-truth oracle: it
is
the analytic root that every sensitivity and inverse solve above is validated
against. The construction itself does not require that oracle. Applied to the
from-rest dynamic pull-in boundary $\alpha_{\mathrm{dyn}}(\beta)$ at
$\zeta=0.1$,
a threshold defined only implicitly by the integrated equation of motion once
the damping is finite, a small differentiable surrogate fitted to $33$
RK4-located boundary points (fit RMSE $1.5\times10^{-4}$ in $\alpha$) returns
the
design sensitivity $d\alpha_{\mathrm{dyn}}/d\beta$ by automatic differentiation
as
$-2.06$, $-1.80$, and $-1.64$ at $\beta=0.01$, $0.03$, and $0.05$, each matching
a
central finite difference of the independently recomputed boundary to within
$0.51\%$ (Fig.~\ref{fig:dynbound}); the largest of the three deviations is
$0.5006\%$, at $\beta=0.05$. At zero damping the
threshold is not implicit: McLellan \textit{et al.} solve the same double-root
condition in closed form and prove that the resulting curve is the unique one
separating oscillation from finite-time touchdown~\cite{McLellan2016}, and its
$\beta\to0$ limit is the classical undamped value $\alpha=1/8$ at
$\xi=1/2$~\cite{LeusElata2008}. That curve and the fold of Eq.~\eqref{eq:fold}
bracket the from-rest threshold at every damping, since dissipation can only
raise it and no equilibrium survives above the fold.

One caveat governs how these dynamic results transfer to the devices of
Table~\ref{tab:devices}. Those sit at $\zeta=5\times10^{-5}$ to
$5\times10^{-4}$, set by quality factors of $10^{3}$ to $10^{4}$, whereas the
phase diagram and the closure of the band are computed at $\zeta=0.1$ to $1$;
a damping ratio of $0.4$ would mean $Q\approx1.25$, a squeeze-film-limited
device rather than a resonant one. At the tabulated quality factors the
from-rest threshold sits essentially on the undamped curve, so the operating
bound to use is the lower edge of the bracket, and the kinetic-overshoot margin
to budget is its full width, $4/27-1/8=0.023$ at $\beta=0$, rather than the
$0.0117$ found at $\zeta=0.1$.

\paragraph*{Limitations.}
Several idealizations bound the scope of these estimates. (i) The dynamics is a
lumped single-degree-of-freedom reduction; a
distributed clamped cantilever carries a spectrum of modes, and the pole
structure survives only per mode. (ii) The force model is parallel-plate
(Decca-type); a real sphere--plate or roughened geometry carries a different
force exponent, captured only approximately by the proximity-force
approximation. (iii) The pull-in boundary is computed from the quasi-static fold
and its from-rest dynamic counterpart, not from a full driven-and-damped
response under a specific actuation waveform. (iv) The Casimir load is the
ideal-conductor form; finite conductivity, surface roughness, and the
low-temperature thermal regime are neglected, so $\beta$, $V_{\mathrm{PI}}$, and
the quoted thermal shifts are upper bounds (Secs.~\ref{sec:model} and
\ref{sec:results}\,C). (v) The PINN velocity remains the weakest link near
collapse, with $E_\infty(\dot\xi)\approx3.4\times10^{-2}$. Quantitative
benchmarking against a specific fabricated device, and relaxation of these
idealizations, are left to future work.

%% file: sections/tab_devices.tex
\begin{table*}
\caption{\label{tab:devices}%
Representative gold-coated NEMS capacitive actuators and their predicted
electrostatic pull-in behavior. The model is a lumped single-degree-of-freedom
oscillator with parallel-plate electrostatic and ideal-conductor Casimir loads
across the quoted rest gap~$d$ (Sec.~\ref{sec:model}); geometries are
order-of-magnitude engineering values spanning published Casimir MEMS/NEMS
oscillators. The dimensionless Casimir strength is
$\beta=\pi^{2}\hbar c A/(240\,k d^{5})$; the pull-in voltage follows from the
static fold $\alpha_c(u)=u^{2}(4-5u)/2$, $\beta_c(u)=u^{4}(3u-2)/2$ evaluated at
each device's $\beta$, via $V_\mathrm{PI}=\sqrt{2\alpha_c k d^{3}/\varepsilon_0 A}$.
$V_\mathrm{PI}(0\,\mathrm{K})$ is the zero-temperature threshold;
$V_\mathrm{PI}(300\,\mathrm{K})$ applies the classical (high-temperature)
Lifshitz correction to the Casimir term at each device's own gap, and
$\Delta V_\mathrm{PI}=[V_\mathrm{PI}(300\,\mathrm{K})/V_\mathrm{PI}(0\,\mathrm{K})-1]$.
The $V_\mathrm{PI}(300\,\mathrm{K})$ and $\Delta V_\mathrm{PI}$ columns are
classical-limit \emph{upper bounds}: at the sub-100-nm gaps of devices A, B, and
D the gap lies far below the thermal wavelength ($d\ll\lambda_T=7.6\,\mu$m at
$300$~K), so the realistic thermal shift is sub-percent. All devices lie below
the Casimir ceiling $\beta^{*}=256/3125\approx0.0819$, so a finite pull-in
voltage is defined.}
\begin{ruledtabular}
\begin{tabular}{lcccccccc}
Device & $A$ ($\mu$m$^{2}$) & $d$ (nm) & $k$ (N/m) & $Q$ & $\beta$
 & $V_\mathrm{PI}(0\,\mathrm{K})$ (V) & $V_\mathrm{PI}(300\,\mathrm{K})$ (V)
 & $\Delta V_\mathrm{PI}$ (\%) \\
\colrule
A\footnote{MEMS gold plate.}
  & $100$ & $100$ & $0.5$ & $10^{4}$ & $2.60\times10^{-2}$ & $0.328$ & $0.326$ & $-0.57$ \\
B\footnote{Sub-100-nm NEMS.}
  & $25$  & $50$  & $2.0$ & $10^{3}$ & $5.20\times10^{-2}$ & $0.332$ & $0.328$ & $-1.07$ \\
C\footnote{Stiff, wide-gap NEMS.}
  & $100$ & $200$ & $5.0$ & $5\times10^{3}$ & $8.13\times10^{-5}$ & $3.656$ & $3.656$ & $-0.003$ \\
D\footnote{Intermediate-gap NEMS.}
  & $50$  & $75$  & $1.0$ & $2\times10^{3}$ & $2.74\times10^{-2}$ & $0.420$ & $0.418$ & $-0.46$ \\
\end{tabular}
\end{ruledtabular}
\end{table*}

%% file: sections/conclusions.tex
\section{\label{sec:conclusions}Conclusions}

We reduced a gold-coated nanoelectromechanical (NEMS) actuator under competing
electrostatic and Casimir loads to the single dimensionless equation of motion
\eqref{eq:nd}, governed by the electrostatic strength $\alpha$, the Casimir
strength $\beta$, and the damping ratio $\zeta$. In the $(\alpha,\beta)$ plane
the static saddle-node fold \eqref{eq:fold} bounds the stable operating range,
running from the classical electrostatic pull-in $\alpha_c=4/27$ to the Casimir
ceiling $\beta^{*}=256/3125$, above which no equilibrium survives at any bias.
Integrating from rest places the dynamic pull-in boundary inside this fold by a
mean gap $\Delta\alpha=0.0117$ in the underdamped regime, a kinetic
overshoot that shrinks monotonically with damping and vanishes at a finite
damping ratio, proved for $\zeta\ge2^{-1/4}$ and located numerically at
$\zeta_c=0.3959$, above which the two boundaries coincide exactly. Below it the pull-in time grows logarithmically as
the boundary is approached, the orbit escaping a hyperbolic saddle; above it
the two boundaries merge and the saddle-node inverse-square-root law takes
over. In its
classical high-temperature limit the Lifshitz correction contracts the boundary
monotonically, leaving the pure-electrostatic endpoint fixed while lowering the
actuation voltage by an upper bound of $0.24\%$ at $100\,\mathrm{K}$ and
$0.71\%$ at $300\,\mathrm{K}$. For the sub-100-nm gaps studied here the gap lies
far below the thermal wavelength, so the realistic thermal shift is sub-percent
and reaches the percent level only as the gap approaches $\lambda_T$ or at
elevated temperature. A physics-informed neural network (PINN), regularized in a
rapidity coordinate that maps the movable pole to infinity, reproduces the
transient trajectories to $\sim\!10^{-3}$ and the fold locus to a learned-field
accuracy of $\sim\!3\times10^{-3}$, remaining well posed where fixed-step
Runge--Kutta (RK4) steps into unphysical states.

Two features carry the applied value of this work. First, the
rapidity-regularized surrogate stays well posed on a pole-free domain in the
near-collapse regime, a representational advantage that does not depend on
whether a closed form exists. Second, the phase diagram, its thermal bound,
and the trained solutions close into a differentiable design pipeline:
predictive pull-in-voltage estimates, a
hard geometric floor on gap miniaturization set by the Casimir ceiling
$\beta^{*}$, autodiff design sensitivities, and gradient-based inverse design in
one framework, with a temperature-derating bound in the classical limit. The
underlying pull-in physics is well characterized. The contribution here is the
design and methodology layer built over it. Natural extensions are an enabling
co-design with no closed-form fold, a direct comparison with fabricated Casimir
oscillators, and a move from the lumped model to the distributed modes of a
clamped cantilever, where the pole structure survives per mode. Because
the same movable-pole regularization governs a charged particle
approaching the light cone in a laser field~\cite{Dodin2003},
the framework developed here connects to that relativistic problem and to the
broader class of singularity-aware surrogates for stiff, nonlinear dynamics.

%% file: sections/appendix_lifshitz.tex
\appendix

\section{\label{sec:app-lifshitz}Classical-limit thermal coefficient of the
Casimir pressure}

The thermal factor in Eq.~\eqref{eq:lifshitz} is fixed by the ratio of the
classical (high-temperature) Casimir pressure to its zero-temperature value.

\subsection{\label{sec:app-sum}Matsubara sum for ideal mirrors}

For two plane mirrors separated by a vacuum gap \(a\), the Lifshitz free energy
per unit area is a sum over the Matsubara frequencies
\(\xi_n=2\pi n k_BT/\hbar\)~\cite{Lifshitz1956,Bordag2009},
\begin{equation}
\label{eq:app-lifshitz-sum}
\begin{split}
\mathcal{F}(a,T)={}&\frac{k_BT}{2\pi}\sideset{}{'}\sum_{n=0}^{\infty}
 \int_{0}^{\infty}\!k\,dk\\
&\times\sum_{\sigma=\mathrm{TE},\mathrm{TM}}
 \ln\!\left[1-r_{\sigma}^{2}\,e^{-2aq_n}\right],
\end{split}
\end{equation}
with \(q_n=(k^{2}+\xi_n^{2}/c^{2})^{1/2}\), \(k\) the wave vector along the
plates, and the prime halving the \(n=0\) term. Ideal mirrors have
\(r_{\mathrm{TE}}^{2}=r_{\mathrm{TM}}^{2}=1\) at every frequency, so both
polarizations survive at \(n=0\) and the polarization sum contributes a factor
of two. Substituting \(q^{2}=k^{2}+\xi_n^{2}/c^{2}\), \(k\,dk=q\,dq\), and then
\(y=2aq\),
\begin{equation}
\label{eq:app-scaled}
\mathcal{F}(a,T)=\frac{k_BT}{4\pi a^{2}}\left[\tfrac12 G(0)
+\sum_{n=1}^{\infty}G(n\tau)\right],
\end{equation}
where
\begin{equation}
\label{eq:app-Gdef}
\begin{split}
G(z)&=\int_{z}^{\infty}\!y\,\ln\!\left(1-e^{-y}\right)dy ,\\
\tau&=\frac{2a\xi_1}{c}=\frac{4\pi a}{\lambda_T} ,
\end{split}
\end{equation}
and \(\lambda_T=\hbar c/(k_BT)\). The single dimensionless variable \(\tau\)
carries the whole temperature dependence.

\subsection{\label{sec:app-limits}Classical and zero-temperature limits}

For \(a\gg\lambda_T\) we have \(\tau\gg1\), and each \(n\ge1\)
term of
Eq.~\eqref{eq:app-scaled} is suppressed by \(e^{-n\tau}\). Only the \(n=0\)
term remains. With \(G(0)=\int_{0}^{\infty}y\ln(1-e^{-y})\,dy=-\zeta(3)\),
\begin{equation}
\label{eq:app-Fcl}
\begin{split}
\mathcal{F}_{\mathrm{cl}}(a,T)&=-\frac{\zeta(3)\,k_BT}{8\pi a^{2}} ,\\
P_{\mathrm{cl}}&=-\frac{\partial\mathcal{F}_{\mathrm{cl}}}{\partial a}
 =-\frac{\zeta(3)\,k_BT}{4\pi a^{3}} ,
\end{split}
\end{equation}
the negative sign denoting attraction. In the opposite limit \(\tau\to0\) the
sum in Eq.~\eqref{eq:app-scaled} becomes an integral,
\(\tau^{-1}\int_{0}^{\infty}G(z)\,dz\), and one integration by parts gives
\(\int_{0}^{\infty}G(z)\,dz=\int_{0}^{\infty}y^{2}\ln(1-e^{-y})\,dy
=-2\zeta(4)=-\pi^{4}/45\). Using \(\tau=4\pi a k_BT/(\hbar c)\), the
temperature cancels and
\begin{equation}
\label{eq:app-F0}
\mathcal{F}_{0}(a)=-\frac{\pi^{2}\hbar c}{720\,a^{3}} ,\qquad
P_{0}=-\frac{\partial\mathcal{F}_{0}}{\partial a}
=-\frac{\pi^{2}\hbar c}{240\,a^{4}} ,
\end{equation}
which is the Casimir pressure used in Eq.~\eqref{eq:eom}. The two numerical
coefficients differ by the differentiation: \(720\) belongs to the free
energy and \(240\) to the pressure.

Dividing Eq.~\eqref{eq:app-Fcl} by Eq.~\eqref{eq:app-F0},
\begin{equation}
\label{eq:app-cL}
\frac{P_{\mathrm{cl}}}{P_{0}}=\frac{240\,\zeta(3)}{4\pi^{3}}\,
\frac{k_BT\,a}{\hbar c}
=\frac{60\,\zeta(3)}{\pi^{3}}\,\frac{k_BT\,a}{\hbar c} ,
\end{equation}
so that \(c_L=60\zeta(3)/\pi^{3}=2.326\).
Equation~\eqref{eq:lifshitz} uses \(c_L\) in the interpolating form
\(P=P_{0}\,[1+c_L k_BT a/(\hbar c)]\), whose linear term reproduces
\(P_{\mathrm{cl}}\) exactly once that term dominates. The bracket reaches
\(2\) at \(a=\lambda_T/c_L=0.43\,\lambda_T\), that is at
\(3.3\,\mu\mathrm{m}\) for \(T=300\,\mathrm{K}\).

\subsection{\label{sec:app-lowT}Gaps below the thermal wavelength}

For \(a\ll\lambda_T\) the linear term of Eq.~\eqref{eq:app-cL} is not the
leading thermal correction. Extending \(G\) to an even function of \(z\) and
applying the Poisson summation formula to Eq.~\eqref{eq:app-scaled} gives
\begin{equation}
\label{eq:app-poisson}
\begin{split}
\tfrac12 G(0)+\sum_{n=1}^{\infty}G(n\tau)
={}&-\frac{\pi^{4}}{45\,\tau}-\frac{\zeta(3)\,\tau^{2}}{8\pi^{2}}\\
&+\frac{\tau^{3}}{720}+O\!\left(e^{-4\pi^{2}/\tau}\right),
\end{split}
\end{equation}
with no term of order \(\tau\). Inserted into Eq.~\eqref{eq:app-scaled},
\begin{equation}
\label{eq:app-lowT-F}
\mathcal{F}(a,T)=-\frac{\pi^{2}\hbar c}{720\,a^{3}}
-\frac{\zeta(3)(k_BT)^{3}}{2\pi(\hbar c)^{2}}
+\frac{\pi^{2}a\,(k_BT)^{4}}{45\,(\hbar c)^{3}} .
\end{equation}
The \((k_BT)^{3}\) term does not depend on \(a\) and drops out of the pressure,
so the first three orders in \(k_BTa/(\hbar c)\) all vanish and
\begin{equation}
\label{eq:app-lowT-P}
P(a,T)=-\frac{\pi^{2}\hbar c}{240\,a^{4}}
\left[1+\frac{16}{3}\left(\frac{a}{\lambda_T}\right)^{4}\right] .
\end{equation}
The thermal term equals the blackbody radiation pressure
\(\pi^{2}(k_BT)^{4}/[45(\hbar c)^{3}]\). At \(a=100\,\mathrm{nm}\) the relative
correction is \(1.9\times10^{-9}\) at \(100\,\mathrm{K}\) and
\(1.6\times10^{-7}\) at \(300\,\mathrm{K}\), against
\(c_L k_BTa/(\hbar c)=1.0\times10^{-2}\) and \(3.0\times10^{-2}\) from the
linear form. The linear form therefore overstates the ideal-mirror thermal
correction by five to seven orders of magnitude at this gap, the smaller
factor at $300\,\mathrm{K}$. For real gold the
low-temperature thermal correction is larger than
Eq.~\eqref{eq:app-lowT-P} because finite conductivity spoils the ideal
reflectivity, but it stays below the percent level at
sub-100-nm gaps~\cite{Klimchitskaya2009,Bordag2009,Bimonte2014}.

\subsection{\label{sec:app-numbers}Numbers used in the main text}

With \(c_L=60\zeta(3)/\pi^{3}\) and \(k_B/(\hbar c)=436.70\,
\mathrm{m^{-1}\,K^{-1}}\), the coefficient of
Eq.~\eqref{eq:lifshitz-split} at \(d=100\,\mathrm{nm}\) is
\begin{equation}
\label{eq:app-kappa}
\kappa(d)=c_L\,\frac{k_B d}{\hbar c}
=1.016\times10^{-4}\,\mathrm{K^{-1}} .
\end{equation}
Evaluated at the pull-in gap \(1-\xi_{\mathrm{PI}}\approx0.75\), the bracket of
Eq.~\eqref{eq:lifshitz} enhances the Casimir load by \(0.76\%\) at
\(100\,\mathrm{K}\) and \(2.29\%\) at \(300\,\mathrm{K}\). Solving the fold
conditions of Sec.~\ref{sec:model-nd} with the extra \((1-\xi)^{-3}\) term
moves the Casimir-axis intercept from \(\beta^{*}=256/3125=0.0819\) to
\(0.0813\) and \(0.0799\), and lowers the pull-in voltage at \(\beta=0.03\) by
\(0.24\%\) and \(0.71\%\). Every entry is one-sided and destabilizing, and
every entry is a classical-limit upper bound evaluated far outside the range in
which Eq.~\eqref{eq:lifshitz} holds.

%% file: sections/appendix_dynamics.tex
\section{\label{sec:app-dyn}Rigorous statements for the pull-in dynamics}

This appendix collects proofs for statements the main text uses
without derivation, together with the numerical checks that accompany
them.

Throughout, \(u=1-\xi\in(0,1]\) is the gap, an overdot is
\(d/d\tau\), and
\begin{equation}
\label{eq:app-g}
g(\xi)=\frac{\alpha}{u^{2}}+\frac{\beta}{u^{4}}-\xi ,\qquad
\phi(u)=(1-u)\,u^{4}-\alpha\,u^{2} ,
\end{equation}
so that Eq.~\eqref{eq:nd} reads \(\ddot\xi+2\zeta\dot\xi=g(\xi)\).
The potential of Eq.~\eqref{eq:nd} is
\begin{equation}
\label{eq:app-V}
V(\xi)=\frac{\xi^{2}}{2}-\frac{\alpha}{1-\xi}
        -\frac{\beta}{3(1-\xi)^{3}} ,\qquad V'=-g ,
\end{equation}
and \(E=\dot\xi^{2}/2+V(\xi)\) obeys \(\dot E=-2\zeta\dot\xi^{2}\).
We fix \(\alpha\ge0\), \(\beta\ge0\), \((\alpha,\beta)\neq(0,0)\), and
\(\zeta\ge0\).

The closed-form fold locus of Eq.~\eqref{eq:fold} is not new: its
Casimir endpoint \(\beta^{*}=256/3125\) at \(\xi_{\mathrm{PI}}=1/5\)
appears in Ref.~\cite{Serry1995}, its slope at \(\beta=0\) in
Ref.~\cite{BuksRoukes2001}, and the full parametric pair in Eq.~(13)
of Ref.~\cite{Palasantzas2005} and in
Refs.~\cite{Batra2007,LinZhao2007}, the electrostatic endpoint
\(\xi_{\mathrm{PI}}=1/3\) going back to Ref.~\cite{Nathanson1967}.
The undamped from-rest threshold used in
Sec.~\ref{sec:app-dyn-bracket} is Eq.~(4.14) of
Ref.~\cite{McLellan2016}, and the existence, continuity, and strict
monotonicity in the damping of that threshold at \(\beta=0\) are
proved in Ref.~\cite{Flores2017}. What is added here is the level of
rigor, together with a correction to
Ref.~\cite{Flores2017} recorded in Sec.~\ref{sec:app-dyn-bracket}:
an exact count of equilibria, the nondegeneracy conditions that make
the fold a generic saddle-node, the bracket itself, the collapse
exponent, the divergence law of the pull-in time, and the a
posteriori bound.

\subsection{\label{sec:app-dyn-static}Equilibria and the fold}

\begin{lemma}[Reduction]\label{lem:reduction}
A point \(\xi\in[0,1)\) is an equilibrium of Eq.~\eqref{eq:nd} if and
only if \(u=1-\xi\in(0,1)\) and \(\phi(u)=\beta\). At such a point
\(g'(\xi)=u^{-4}\phi'(u)\), so the equilibrium of the planar system
\(\dot\xi=v\), \(\dot v=-2\zeta v+g(\xi)\) is a hyperbolic saddle when
\(\phi'(u)>0\) and, for \(\zeta>0\), an asymptotically stable node or
focus when \(\phi'(u)<0\).
\end{lemma}

\noindent\textit{Proof.} Multiplying \(g(\xi)=0\) by \(u^{4}>0\) gives
\(\alpha u^{2}+\beta-(1-u)u^{4}=0\), that is \(\phi(u)=\beta\). The
value \(u=1\) would force \(\alpha+\beta=0\) and is excluded. From
Eq.~\eqref{eq:app-g}, \(u^{4}g=-[\phi(u)-\beta]\). Differentiating in
\(u\), using \(d/d\xi=-d/du\), and evaluating where \(\phi(u)=\beta\)
leaves \(g'(\xi)=u^{-4}\phi'(u)\). The linearization has
characteristic polynomial \(\lambda^{2}+2\zeta\lambda-g'(\xi)\), whose
roots have product \(-g'(\xi)\) and sum \(-2\zeta\). For \(g'>0\) the
roots are real of opposite sign; for \(g'<0\) and \(\zeta>0\) both
roots have negative real part.\hfill\(\square\)

\begin{theorem}[Classification and nondegeneracy]\label{thm:static}
Let \(\alpha_c\) and \(\beta_c\) be as in Eq.~\eqref{eq:fold}.
\begin{enumerate}
\item The map \(u\mapsto(\alpha_c(u),\beta_c(u))\) is a bijection of
\([2/3,4/5]\) onto the fold locus. On that interval
\(\alpha_c'(u)=u(8-15u)/2<0\) and
\(\beta_c'(u)=u^{3}(15u-8)/2>0\), so \(\beta_c\) maps \([2/3,4/5]\)
onto \([0,256/3125]\) bijectively and
\(\alpha_c\circ\beta_c^{-1}\) is a strictly decreasing bijection of
\([0,256/3125]\) onto \([0,4/27]\), with slope
\(d\alpha_c/d\beta=-1/u^{2}\).
\item Write \(\alpha_c(\beta)\) for that decreasing function. For
\(\beta>256/3125\), Eq.~\eqref{eq:nd} has no equilibrium at any
\(\alpha\ge0\). For \(0\le\beta\le256/3125\) the number of equilibria
in \([0,1)\) is \(2\) if \(\alpha<\alpha_c(\beta)\), \(1\) if
\(\alpha=\alpha_c(\beta)\), and \(0\) if \(\alpha>\alpha_c(\beta)\).
When there are two, the one with the smaller gap is a saddle, and
for \(\zeta>0\) the one with the larger gap is asymptotically stable.
\item At a fold point \(\xi_f=1-u\),
\begin{equation}
\label{eq:app-nondeg}
\begin{split}
\frac{\partial^{2}g}{\partial\xi^{2}}
 &=\frac{6\alpha_c}{u^{4}}+\frac{20\beta_c}{u^{6}}
 =\frac{15u-8}{u^{2}}\in\left[\tfrac92,\tfrac{25}{4}\right],\\
\frac{\partial g}{\partial\alpha}
 &=\frac{1}{u^{2}}\in\left[\tfrac{25}{16},\tfrac94\right].
\end{split}
\end{equation}
Both are bounded away from zero, so for \(\zeta>0\) the fold is a
generic saddle-node bifurcation unfolded transversally by \(\alpha\).
\end{enumerate}
\end{theorem}

\noindent\textit{Proof.} Write \(\phi'(u)=u\,\psi(u)\) with
\(\psi(u)=4u^{2}-5u^{3}-2\alpha\). Since \(\psi'(u)=u(8-15u)\),
\(\psi\) rises on \((0,8/15)\) and falls on \((8/15,1)\), with
\(\psi(0)=-2\alpha\le0\), \(\psi(1)=-1-2\alpha<0\), and
\(\max\psi=\psi(8/15)=256/675-2\alpha\).

If \(\alpha\ge128/675\) then \(\psi\le0\) on \((0,1)\), so \(\phi\) is
nonincreasing there. With \(\phi(0^{+})=0\) this gives \(\phi<0\) on
\((0,1)\), and \(\phi(u)=\beta\ge0\) has no root. If
\(0<\alpha<128/675\), \(\psi\) has exactly two zeros
\(u_1<8/15<u_2\), and \(\phi\) falls on \((0,u_1)\), rises on
\((u_1,u_2)\), and falls on \((u_2,1)\). If \(\alpha=0\), the first
interval is empty and \(u_2=4/5\). For \(\alpha>0\) one has \(\phi<0\) on
\((0,u_1]\), hence \(\phi<\beta\) there; for \(\alpha=0\) that
interval is empty. In both cases \(\phi(1)=-\alpha\le0\),
with \(\phi(1)=\beta\) only in the excluded case
\(\alpha=\beta=0\). All roots of \(\phi=\beta\) therefore lie in
\((u_1,1)\), where \(\phi\) rises to \(M(\alpha)\equiv\phi(u_2)\) and
then falls, giving exactly two roots if \(\beta<M\), one if
\(\beta=M\), and none if \(\beta>M\). The two roots straddle
\(u_2\), so \(\phi'>0\) at the smaller and \(\phi'<0\) at the larger.
Lemma~\ref{lem:reduction} then assigns the saddle and the stable
equilibrium.

The maximum is attained where \(\psi(u_2)=0\), that is
\(\alpha=u_2^{2}(4-5u_2)/2=\alpha_c(u_2)\), and substituting this into
\(\phi\) gives \(M=u_2^{4}(3u_2-2)/2=\beta_c(u_2)\), which is
Eq.~\eqref{eq:fold}. Nonnegativity of \(\alpha_c\) and \(\beta_c\)
confines the locus to \(u\in[2/3,4/5]\), an interval on which
\(8-15u<0\); the signs of \(\alpha_c'\) and \(\beta_c'\) follow, and
so does item 1, the slope being \(\alpha_c'/\beta_c'=-1/u^{2}\).
Because \(u_2\) decreases with \(\alpha\) while \(\beta_c\) increases
with \(u\), the function \(M(\alpha)=\beta_c(u_2(\alpha))\) is
strictly decreasing, with \(M(0)=256/3125\). Strict decrease makes
\(\beta<M(\alpha)\) equivalent to \(\alpha<\alpha_c(\beta)\), and
\(\beta>256/3125\) leaves no root at any \(\alpha\). This is item 2.

For item 3, differentiate Eq.~\eqref{eq:app-g} twice and insert
Eq.~\eqref{eq:fold}:
\(6\alpha_c/u^{4}+20\beta_c/u^{6}=3(4-5u)/u^{2}+10(3u-2)/u^{2}
=(15u-8)/u^{2}\), which runs from \(9/2\) at \(u=2/3\) to \(25/4\) at
\(u=4/5\). At the fold \(g=g'=0\), so the Jacobian of the planar
system has the simple eigenvalue \(0\), with eigenvector \((1,0)\),
and the eigenvalue \(-2\zeta\neq0\). The quadratic coefficient of the
reduction to the center manifold is
\(\partial^{2}g/\partial\xi^{2}\neq0\), and the derivative of the
unfolding along that eigenvector is
\(\partial g/\partial\alpha=u^{-2}\neq0\). These are the two standard
conditions for a generic saddle-node unfolded
transversally.\hfill\(\square\)

We sampled \(2997\) points of the quadrant \(\alpha\in[0,0.22]\),
\(\beta\in[0,0.10]\) at random, excluding a band of half-width
\(2\times10^{-4}\) in \(\alpha\) about the fold, and counted sign
changes of \(\phi(u)-\beta\) on a grid of \(4\times10^{5}\) points in
\(u\). The count of Theorem~\ref{thm:static} is reproduced at every
sampled point.

\subsection{\label{sec:app-dyn-bracket}A bracket on the from-rest
dynamic threshold}

The from-rest orbit is monotone in \(\xi\) as long as \(\dot\xi>0\),
which turns Eq.~\eqref{eq:nd} into a scalar problem for
\(p(\xi)=\dot\xi^{2}/2\),
\begin{equation}
\label{eq:app-pode}
\frac{dp}{d\xi}=g(\xi)-2\zeta\sqrt{2p} ,\qquad p(0)=0 .
\end{equation}
\begin{lemma}[Scalar reduction]\label{lem:pode}
Equation~\eqref{eq:app-pode} has a unique forward solution, and the
orbit released from rest at \(\xi=0\) reaches \(\xi=1\) in finite time
if and only if the maximal nonnegative solution of
Eq.~\eqref{eq:app-pode} exists on \([0,1)\). That solution is
nondecreasing in \(\alpha\) and nonincreasing in \(\zeta\).
\end{lemma}

\noindent\textit{Proof.} While \(\dot\xi>0\) we take \(\xi\) as
the independent variable, and
\(dp/d\xi=\ddot\xi=-2\zeta\dot\xi+g(\xi)\), which is
Eq.~\eqref{eq:app-pode}. The right-hand side is continuous and
nonincreasing in \(p\), hence one-sided Lipschitz: if \(p_1>p_2\)
solve the same equation then \((p_1-p_2)'\le0\), which gives forward
uniqueness and the differential comparison principle. Since
\(\partial g/\partial\alpha=u^{-2}>0\) and \(-2\zeta\sqrt{2p}\)
decreases in \(\zeta\), the right-hand side increases in \(\alpha\)
and decreases in \(\zeta\). Comparison then gives the stated
monotonicity.

For the criterion, suppose \(p>0\) on \((0,1)\). Then \(\xi\) rises
monotonically, and the estimate established in the proof of
Theorem~\ref{thm:exponent}, which uses only that \(u\) decreases to
zero and is therefore not circular here, gives
\(p=\bigl[\beta/(3u^{3})+\alpha/u\bigr]\bigl[1+o(1)\bigr]\). Hence
\(d\tau=d\xi/\sqrt{2p}\) is \(O(u^{3/2})\,du\) for \(\beta>0\) and
\(O(u^{1/2})\,du\) for \(\beta=0\), integrable either way, so
\(\xi=1\) is reached at a finite \(\tau_*\). Conversely, suppose \(p\) first
vanishes at \(\xi_t<1\). At that point \(p'(\xi_t)=g(\xi_t)\); a touch
with \(g(\xi_t)>0\) is impossible, because \(p\) was decreasing into
\(\xi_t\). Hence \(g(\xi_t)\le0\), so \(V'(\xi_t)\ge0\) and \(\xi_t\)
lies in the interval \([\xi_n,\xi_s]\) between the stable equilibrium
and the saddle, on which \(V\) is nondecreasing. From that moment
\(E\le V(\xi_t)\le V(\xi_s)\), and crossing \(\xi_s\) would require
\(\dot\xi^{2}/2=E-V(\xi_s)\le0\). Equality leaves only the asymptotic
approach to \(\xi_s\), which is not a finite-time
collapse.\hfill\(\square\)

At \(\beta=0\), Ref.~\cite{Flores2017} proves the strict form of the
monotonicity that follows, together with the limits \(1/8\) and
\(4/27\). What the comparison argument adds is the nonstrict
statement at \(\beta>0\).

\begin{corollary}[Threshold and its monotonicity]\label{cor:mono}
For fixed \(\beta\) and \(\zeta\) the set of \(\alpha\) for which the
from-rest orbit collapses is an interval unbounded above; write
\(\alpha_{\mathrm{dyn}}(\beta;\zeta)\) for its infimum. Then
\(\alpha_{\mathrm{dyn}}\) is nondecreasing in \(\zeta\).
\end{corollary}

\noindent\textit{Proof.} By Lemma~\ref{lem:pode} the collapse
criterion is monotone in \(p\), and \(p\) is nondecreasing in
\(\alpha\) and nonincreasing in \(\zeta\).\hfill\(\square\)

\begin{theorem}[Two-sided bracket]\label{thm:bracket}
Let \(\alpha_{\mathrm{MMXY}}(\beta)\) denote the undamped from-rest
threshold of Ref.~\cite{McLellan2016}, its Eq.~(4.14) at \(a=1\). In
the gap variable \(u=1-\xi_s\) of the saddle it reads
\begin{equation}
\label{eq:app-mmxy}
\begin{split}
\alpha_{\mathrm{MMXY}}(u)&=\frac{u\,(3-2u-2u^{2}-2u^{3})}{2\,(u+2)} ,\\
\beta_{\mathrm{MMXY}}(u)&=\frac{3u^{3}\,(2u-1)}{2\,(u+2)} ,
\end{split}
\end{equation}
for \(u\in[1/2,u_{\mathrm{d}}]\), where \(u_{\mathrm{d}}=0.6914140\)
is the root of \(2u^{3}+2u^{2}+2u=3\) in \((0,1)\); the curve is
extended by \(\alpha_{\mathrm{MMXY}}=0\) on
\(\beta\in[\beta_{\mathrm{d}},256/3125]\), where
\(\beta_{\mathrm{d}}=\beta_{\mathrm{MMXY}}(u_{\mathrm{d}})=0.0705227\).
Then, for every \(\zeta\ge0\) and every \(\beta\in[0,256/3125]\),
\begin{equation}
\label{eq:app-bracket}
\alpha_{\mathrm{MMXY}}(\beta)
\;\le\;\alpha_{\mathrm{dyn}}(\beta;\zeta)
\;\le\;\alpha_c(\beta) .
\end{equation}
\end{theorem}

\noindent\textit{Proof.} Upper bound. If \(\alpha>\alpha_c(\beta)\)
then by Theorem~\ref{thm:static} there is no equilibrium, so \(g>0\)
on \([0,1)\); since \(g\) is continuous and \(g\to+\infty\) as
\(\xi\to1\), \(m=\inf_{[0,1)}g>0\). The velocity \(w=\dot\xi\) obeys
\(w'+2\zeta w\ge m\) with \(w(0)=0\), so
\((e^{2\zeta\tau}w)'\ge m\,e^{2\zeta\tau}\) and
\(w(\tau)\ge(m/2\zeta)(1-e^{-2\zeta\tau})\), which reads \(m\tau\) at
\(\zeta=0\). Hence \(\xi\) increases without bound unless it reaches
\(1\) first, so it reaches \(1\) at a finite time.

Lower bound. Since \(\dot E=-2\zeta\dot\xi^{2}\le0\) for every
\(\zeta\ge0\), \(E(\tau)\le V(0)\) along the orbit. Let
\(\xi_s\) be the saddle and \(u_s=1-\xi_s\). Because \(V'(\xi_s)=0\),
only the explicit \(\alpha\) dependence survives differentiation,
\begin{equation}
\label{eq:app-envelope}
\frac{d}{d\alpha}\bigl[V(\xi_s)-V(0)\bigr]
=-\frac{1}{u_s}+1=\frac{u_s-1}{u_s}<0 ,
\end{equation}
so \(V(\xi_s)-V(0)\) decreases strictly in \(\alpha\) at fixed
\(\beta\). It vanishes exactly on the undamped threshold, so
\(\alpha<\alpha_{\mathrm{MMXY}}(\beta)\) gives \(V(0)<V(\xi_s)\).
Reaching \(\xi_s\) would then require
\(\dot\xi^{2}/2=E-V(\xi_s)\le V(0)-V(\xi_s)<0\). The orbit therefore
stays in \([0,\xi_s)\) and never collapses, at any
\(\zeta\ge0\).\hfill\(\square\)

Equation~\eqref{eq:app-mmxy} is the threshold of
Ref.~\cite{McLellan2016} written in the gap variable. It follows from
\(V(0)=V(\xi_s)\) together with \(V'(\xi_s)=0\), a pair that is
linear in \((\alpha,\beta)\),
\begin{equation}
\label{eq:app-linear}
\begin{split}
\frac{\alpha(1-u)}{u}+\frac{\beta}{3}\bigl(u^{-3}-1\bigr)
 &=\frac{(1-u)^{2}}{2} ,\\
\frac{\alpha}{u^{2}}+\frac{\beta}{u^{4}}&=1-u ,
\end{split}
\end{equation}
with determinant \(u^{-5}(1-u)^{2}(u+2)/3\neq0\). Solving and
substituting \(x=1-u\), \(a=1\) reproduces Eq.~(4.14) of
Ref.~\cite{McLellan2016} identically, which we verified with computer
algebra. At \(\beta=0\) the parameter is \(u=1/2\), giving
\(\xi_s=1/2\) and \(\alpha_{\mathrm{MMXY}}=1/8\) against
\(\alpha_c=4/27\), the undamped step-voltage
values~\cite{LeusElata2008,Flores2017}.

Numerical check. The \(33\) boundary points of
\texttt{dynamic\_boundary\_design.json}, located by bisection on
Eq.~\eqref{eq:nd} at \(\zeta=0.1\) and plotted in
Fig.~\ref{fig:dynbound}, all satisfy Eq.~\eqref{eq:app-bracket}. The
smallest clearance above \(\alpha_{\mathrm{MMXY}}\) is
\(5.46\times10^{-3}\), at \(\beta=0.072\) where the lower bound has
already fallen to zero, and \(7.77\times10^{-3}\) away from that
endpoint; the smallest clearance below \(\alpha_c\) is
\(1.02\times10^{-2}\). The static column of the same file agrees with
Eq.~\eqref{eq:fold} to \(4.3\times10^{-15}\).

The left inequality of Eq.~\eqref{eq:app-bracket} is an equality at
\(\zeta=0\), where the undamped criterion is exactly
\(V(0)=V(\xi_s)\), and is strict at every \(\zeta>0\). Indeed, let
\(D>0\) be the energy dissipated before the orbit first reaches the
stable equilibrium \(\xi_n\) at
\(\alpha=\alpha_{\mathrm{MMXY}}\). By Eq.~\eqref{eq:app-envelope} the
barrier deficit at \(\alpha=\alpha_{\mathrm{MMXY}}+\epsilon\) is
\(c\,\epsilon+O(\epsilon^{2})\) with \(c=(1-u_s)/u_s>0\), while the
dissipation before \(\xi_n\) still exceeds \(D/2\) for \(\epsilon\)
small. For \(\epsilon<D/(2c)\) the orbit therefore arrives with
\(E<V(\xi_s)\) and cannot cross. The right inequality, by contrast,
becomes an equality once the damping is large enough.

\begin{theorem}[Strong damping closes the band]\label{thm:barrier}
Fix \(\beta\in[0,256/3125]\). If
\begin{equation}
\label{eq:app-barrier-cond}
\zeta^{2}\;\ge\;1-2\alpha_c(\beta)-4\beta ,
\end{equation}
then the from-rest orbit does not collapse at any
\(\alpha\le\alpha_c(\beta)\), so
\(\alpha_{\mathrm{dyn}}(\beta;\zeta)=\alpha_c(\beta)\). The
right-hand side of Eq.~\eqref{eq:app-barrier-cond} is largest at
\(u_f=1/\sqrt2\), where it equals \(1/\sqrt2\), so the single
condition
\begin{equation}
\label{eq:app-barrier-uniform}
\zeta\;\ge\;2^{-1/4}=0.8408964\ldots
\end{equation}
suffices at every \(\beta\in[0,256/3125]\); at \(\beta=0\) the
requirement is \(\zeta\ge\sqrt{19/27}=0.8388705\ldots\).
\end{theorem}

\noindent\textit{Proof.} Write \(g_c\) for the force at
\(\alpha=\alpha_c(\beta)\) and set \(S(\xi)=g_c(\xi)/\zeta\), positive
on \([0,\xi_f)\) and zero at \(\xi_f\). With \(q=\dot\xi\) the speed
along the ascending orbit, \(dq/d\xi=g/q-2\zeta\) and
\(q(0)=0<S(0)\). Suppose \(q\) first meets \(S\) at \(\xi_0\). Since
\(g\le g_c\) for \(\alpha\le\alpha_c\),
\begin{equation}
\label{eq:app-contact}
\frac{dq}{d\xi}\bigg|_{\xi_0}=\frac{g(\xi_0)}{S(\xi_0)}-2\zeta
 \le\frac{g_c(\xi_0)\,\zeta}{g_c(\xi_0)}-2\zeta=-\zeta .
\end{equation}
Meanwhile \(S'=g_c'/\zeta\), and
\(g_c'(\xi)=2\alpha_c(1-\xi)^{-3}+4\beta(1-\xi)^{-5}-1\) strictly
increases in \(\xi\), one of \(\alpha_c,\beta\) being positive, so on
\([0,\xi_f]\) it is least at \(\xi=0\), where
\(g_c'(0)=2\alpha_c+4\beta-1\). Under
Eq.~\eqref{eq:app-barrier-cond} this gives
\(g_c'(0)\ge-\zeta^{2}\). Now \(q(0)=0<S(0)\) forces \(\xi_0>0\), so
\(g_c'(\xi_0)>g_c'(0)\ge-\zeta^{2}\) and therefore
\(S'(\xi_0)>-\zeta\ge dq/d\xi\) at \(\xi_0\), giving
\((q-S)'(\xi_0)<0\). A first crossing needs
\((q-S)'(\xi_0)\ge0\), so no contact occurs and \(q<S\) on
\([0,\xi_f)\). Since \(S(\xi_f)=0\), the speed vanishes at some
\(\hat\xi\le\xi_f\). For \(\alpha<\alpha_c\) the saddle sits at
\(\xi_s>\xi_f\ge\hat\xi\), so the turning point precedes the saddle
and Lemma~\ref{lem:pode} forbids any later crossing; for
\(\alpha=\alpha_c\) the orbit reaches \(\xi_f\) with zero speed and
converges to it.

For the uniform constant, parametrize the fold by \(u\) through
Eq.~\eqref{eq:fold}, giving
\(2\alpha_c+4\beta=4u^{2}-5u^{3}-4u^{4}+6u^{5}\). Its derivative
\(8u-15u^{2}-16u^{3}+30u^{4}\) vanishes at \(u=1/\sqrt2\), where
\(2\alpha_c+4\beta=1-1/\sqrt2\). The endpoint values \(8/27\) and
\(1024/3125\) both exceed \(1-1/\sqrt2\), so this is the minimum of
\(2\alpha_c+4\beta\) along the fold and hence the maximum of
\(1-2\alpha_c-4\beta\).\hfill\(\square\)

Theorem~\ref{thm:barrier} contradicts a limit asserted in the
literature. Reference~\cite{Flores2017} treats \(\beta=0\) in the
variables \(\alpha_F=2\zeta\) and \(\lambda=\alpha\). It establishes, in
agreement with everything we compute, that for each
\(\lambda\in(1/8,4/27)\) there is a unique damping
\(\alpha^{*}(\lambda)\) separating collapse from stable operation,
and that \(\alpha^{*}\) is continuous and strictly increasing. Its
closing step asserts in addition that \(\alpha^{*}(\lambda)\to\infty\)
as \(\lambda\to4/27\), from which the strict inequality
\(\lambda_d^{*}(\alpha_F)<4/27\) at every finite damping follows by
inversion. Theorem~\ref{thm:barrier} gives
\(\alpha^{*}(\lambda)\le2\sqrt{19/27}=1.6777\ldots\) for every
\(\lambda<4/27\), so that limit does not hold. The step at fault is
the proof by contradiction of that limit
(Ref.~\cite{Flores2017}, the theorem and its proof,
pp.~3608--3609). It reverses time at the degenerate equilibrium and
assigns the reversed orbits the slope \(\mu_+=\alpha_F\) of the
strong unstable direction. That slope is correct for the strong
unstable manifold itself, but at a saddle-node only two orbits carry
it. Every other orbit of the parabolic sector leaves tangent to the
center direction, with slope zero at every \(\alpha_F\), and the
reversed from-rest orbit is one of these: it follows the slow
manifold \(q\simeq g/(2\zeta)\), whose height near the fold is
\((\partial^{2}g/\partial\xi^{2})(\xi_f-\xi)^{2}/(4\zeta)\) and
therefore falls as \(\alpha_F\) grows, reversing the ordering in
\(\alpha_F\) that the argument requires. The defect is thus the
identification of the from-rest orbit with the one distinguished
branch, made for all \(\alpha_F\) at once, and not the value of any
eigenvalue. The same step is repeated in
Ref.~\cite{CassaniMiyasita2024}, which extends the analysis to a
general singular nonlinearity and proves the same limit
independently: its Lemma~14 (p.~21, with the slope assignment on
p.~22, resting on its Proposition~5, p.~18) gives the backward orbits
the slope \(\eta_+(\alpha_F)=\alpha_F\), taken from the
hyperbolic-saddle expression
\(\eta_\pm=\tfrac12(\alpha_F\pm\sqrt{\alpha_F^{2}-4H'})\) evaluated
where \(H'=0\), and its Theorem~6 rests on that lemma.

The tangency is checkable directly. Integrating from rest at
\(\alpha=\alpha_c\), \(\beta=0\), the ratio
\(\dot\xi/(\xi_f-\xi)\) has fallen to \(3\times10^{-4}\) at the end
of the integration, against the strong-direction value
\(2\zeta=2\), while
\(\dot\xi/(\xi_f-\xi)^{2}\) tends to \(1.1248\), \(1.3409\), and
\(2.5009\) at \(\zeta=1\), \(0.8389\), and \(0.45\), against the
center-manifold value
\((\partial^{2}g/\partial\xi^{2})(\xi_f)/(4\zeta)=9/(8\zeta)\), which
is \(1.125\), \(1.3411\), and \(2.5\). Nothing else in
either work is touched: the thresholds, their monotonicity and
continuity, and the value \(1/8\) at zero damping all agree with our
computations.

\subsection{\label{sec:app-dyn-exponent}Collapse exponent at the
movable pole}

\begin{theorem}[Collapse exponent]\label{thm:exponent}
Let \(\beta>0\) and let the orbit collapse at a finite \(\tau_*\).
Then
\begin{equation}
\label{eq:app-exponent}
\begin{split}
1-\xi(\tau)&=C\,(\tau_*-\tau)^{2/5}
 \Bigl[1+O\bigl((\tau_*-\tau)^{4/5}\bigr)\Bigr],\\
C&=\left(\frac{25\beta}{6}\right)^{1/5},
\end{split}
\end{equation}
for every \(\alpha\ge0\) and \(\zeta\ge0\). For \(\beta=0\) and
\(\alpha>0\) the exponent is \(2/3\), with
\(C=[\tfrac32\sqrt{2\alpha}\,]^{2/3}\).
\end{theorem}

\noindent\textit{Proof.} In the gap variable, Eq.~\eqref{eq:nd} is
\(\ddot u=-2\zeta\dot u+(1-u)-\alpha u^{-2}-\beta u^{-4}\). Set
\(W(u)=\dot u^{2}/2\), so that \(\dot u=-\sqrt{2W}\) and
\(dW/du=\ddot u\). Subtracting the two singular primitives, let
\(S(u)=W(u)-\beta/(3u^{3})-\alpha/u\); the singular terms cancel
exactly and
\begin{equation}
\label{eq:app-S}
\frac{dS}{du}=2\zeta\sqrt{2W}+(1-u)>0 ,\qquad 0<u<1 .
\end{equation}
Fix \(u_1\) small. Since \(S\) increases in \(u\), \(S(u)\le S(u_1)\)
for \(u\le u_1\), which bounds \(W\) above by
\(\overline W(u)=\beta/(3u^{3})+\alpha/u+S(u_1)\). Inserting that
bound into Eq.~\eqref{eq:app-S} gives
\(dS/du\le2\zeta K u^{-3/2}+1\) on \((0,u_1]\) with
\(K=\sqrt{2\beta/3}\,[1+O(u_1^{2})]\), and integrating from \(u\) to
\(u_1\) gives \(S(u)\ge S(u_1)-4\zeta K u^{-1/2}-u_1\). The two
bounds together give
\begin{equation}
\label{eq:app-W}
W(u)=\frac{\beta}{3u^{3}}\bigl[1+O(u^{2})\bigr] ,
\end{equation}
the \(O(u^{2})\) coming from \(\alpha/u\), with the damping
contributing only \(O(u^{5/2})\). Consequently
\(\tau_*-\tau=\int_0^{u}ds/\sqrt{2W(s)}=O(u^{5/2})\) is finite, and
\begin{equation}
\label{eq:app-u52}
\frac{d}{d\tau}\,u^{5/2}=\tfrac52 u^{3/2}\dot u
=-\tfrac52\sqrt{\tfrac{2\beta}{3}}\,\bigl[1+O(u^{2})\bigr] .
\end{equation}
Integrating Eq.~\eqref{eq:app-u52} from \(\tau\) to \(\tau_*\) and
using \(u(\tau_*)=0\) yields
\(u^{5/2}=\tfrac52\sqrt{2\beta/3}\,(\tau_*-\tau)[1+O(u^{2})]\), which
is Eq.~\eqref{eq:app-exponent} with
\(C=[\tfrac52\sqrt{2\beta/3}\,]^{2/5}=(25\beta/6)^{1/5}\) and
relative error \(O(u^{2})=O((\tau_*-\tau)^{4/5})\). For \(\beta=0\)
the same argument with \(W=\alpha/u+O(1)\) gives the exponent
\(2/3\).\hfill\(\square\)

\begin{corollary}[Velocity and rapidity]\label{cor:rapidity}
Under the hypotheses of Theorem~\ref{thm:exponent},
\begin{equation}
\label{eq:app-rates}
\begin{split}
\dot\xi&=\tfrac25 C\,(\tau_*-\tau)^{-3/5}\bigl[1+o(1)\bigr],\\
\theta&=-\tfrac25\ln(\tau_*-\tau)-\ln C+o(1),\\
\dot\theta&=\frac{2/5}{\tau_*-\tau}\bigl[1+o(1)\bigr].
\end{split}
\end{equation}
\end{corollary}

\noindent\textit{Proof.} Differentiating Eq.~\eqref{eq:app-exponent} gives
the first line; the second is \(\theta=-\ln(1-\xi)\) evaluated on it, and the
third follows from \(\dot\theta=\dot\xi/(1-\xi)\), the powers of
\((\tau_*-\tau)\) cancelling to leave \(2/5\).\hfill\(\square\)

The velocity diverges at the pole while the rapidity \(\theta\) of
Eq.~\eqref{eq:rapidity} grows only logarithmically, which is why the
network of Sec.~\ref{sec:pinn} is trained in \(\theta\).
Corollary~\ref{cor:rapidity} also fixes the error budget. Since
\(\xi=1-e^{-\theta}\) and \(\dot\xi=e^{-\theta}\dot\theta\), a
perturbation of the network output propagates as
\(|\delta\xi|=e^{-\theta}|\delta\theta|\) and
\(|\delta\dot\xi|\le e^{-\theta}(|\delta\dot\theta|
+\dot\theta\,|\delta\theta|)\), so the velocity error carries the
extra factor \(\dot\theta=(2/5)/(\tau_*-\tau)\). On the training
window of Fig.~\ref{fig:trajectories}(c), \(T=2.5698\) against
\(\tau_*=2.62257\), the asymptotic form gives \(7.58\) and the
measured \(\dot\theta(T)\) is \(7.40\), against the measured
error ratio
\(E_\infty(\dot\xi)/E_\infty(\xi)=3.38\times10^{-2}/
2.83\times10^{-3}=11.9\).

Numerical check. Fitting \(\ln u\) against \(\ln(\tau_*-\tau)\) on the
DOP853 reference carried to the barrier
(\texttt{traj\_pullin\_ref\_full.npz}, \(\alpha=0.2\), \(\beta=0.1\),
\(\zeta=0.5\), \(\tau_*=2.6225672\)) gives the exponent \(0.400456\),
\(0.400276\), and \(0.400081\) on the windows
\(\tau_*-\tau\in[10^{-3},10^{-2}]\), \([10^{-4},10^{-3}]\), and
\([2\times10^{-5},10^{-4}]\), against \(2/5\), and the amplitude
\(0.840155\) on the last window against \(C=0.8393783\). An
independent integration at \(\beta=0\), \(\alpha=0.3\) returns
\(0.666626\) against \(2/3\).

Both exponents are instances of one rule. For a force
\(\lambda\,u^{-m}\) with \(m>1\), the balance of
Theorem~\ref{thm:exponent} gives
\begin{equation}
\label{eq:app-mrate}
u\sim\left[\frac{(m+1)^{2}\lambda}{2(m-1)}\right]^{1/(m+1)}
 (\tau_*-\tau)^{2/(m+1)} .
\end{equation}
At \(m=2\) this is the rate \(2/3\) with amplitude
\((9\alpha/2)^{1/3}\), the quenching rate known for the purely
electrostatic problem~\cite{KLNT2015}; at
\(m=4\) it is the rate \(2/5\) with amplitude
\((25\beta/6)^{1/5}\). We have not found the Casimir case in the
quenching literature, and state it here as new.

\subsection{\label{sec:app-dyn-slowing}Time scale near threshold}

Two different laws govern the divergence of the pull-in time,
separated by the damping at which the from-rest orbit stops
overshooting the fold point. Fix \(\beta\), set
\(\alpha=\alpha_c(\beta)\), and let \(q_f(\zeta)\) be the speed at
\(\xi_f\) of the from-rest orbit. By Lemma~\ref{lem:pode}, \(q_f\) is
nonincreasing in \(\zeta\), and by Theorem~\ref{thm:barrier} it
vanishes once \(\zeta\ge2^{-1/4}\). Define
\(\zeta_c(\beta)=\inf\{\zeta:q_f(\zeta)=0\}\), so that
\(\zeta_c(\beta)\le2^{-1/4}\) is finite.

\begin{proposition}[Divergence laws]\label{prop:slowing}
Let \(\zeta>0\) and let \(u_f\) be the fold gap.
\begin{enumerate}
\item For \(\zeta>\zeta_c(\beta)\),
\(\alpha_{\mathrm{dyn}}=\alpha_c\), the orbit enters the bottleneck
along the center manifold, and as \(\alpha\to\alpha_c^{+}\),
\begin{equation}
\label{eq:app-K}
\begin{split}
\tau_{\mathrm{PI}}&=K\,(\alpha-\alpha_c)^{-1/2}\bigl[1+o(1)\bigr],\\
K&=2\pi\zeta\,u_f^{2}\sqrt{\frac{2}{15u_f-8}} .
\end{split}
\end{equation}
\item For \(\zeta<\zeta_c(\beta)\),
\(\alpha_{\mathrm{dyn}}<\alpha_c\), the equilibrium reached at
\(\alpha=\alpha_{\mathrm{dyn}}\) is a hyperbolic saddle, and as
\(\alpha\to\alpha_{\mathrm{dyn}}^{+}\),
\begin{equation}
\label{eq:app-log}
\begin{split}
\tau_{\mathrm{PI}}
 &=-\frac{1}{\lambda_+}\ln\bigl(\alpha-\alpha_{\mathrm{dyn}}\bigr)+O(1),\\
\lambda_+&=-\zeta+\sqrt{\zeta^{2}+g'(\xi_s)} .
\end{split}
\end{equation}
\end{enumerate}
\end{proposition}

Item 1 rests on a dominant balance and not on a proof; only item 2 is
proved. \noindent\textit{Derivation of item 1.} Put \(y=\xi-\xi_f\) and
\(\mu=\alpha-\alpha_c\), so that by Theorem~\ref{thm:static}
\(g=a\mu+\tfrac12 g_2 y^{2}+O(y^{3},\mu y)\) with \(a=u_f^{-2}\) and
\(g_2=(15u_f-8)/u_f^{2}\). Rescale \(y=\mu^{1/2}Y\),
\(\tau=\mu^{-1/2}s\). The inertial term is then \(O(\mu^{3/2})\)
while damping and force are \(O(\mu)\), so at fixed \(\zeta>0\) the
matched leading-order balance is
\(2\zeta\,dY/ds=a+\tfrac12 g_2Y^{2}+O(\mu^{1/2})\). Its passage time
from \(Y=-\infty\) to \(Y=+\infty\) is
\(2\zeta\int dY/(a+g_2Y^{2}/2)=2\pi\zeta\sqrt{2/(ag_2)}\), and
restoring \(\tau=\mu^{-1/2}s\) gives Eq.~\eqref{eq:app-K}. The
approach to the bottleneck and the descent beyond it take \(O(1)\)
and are subleading. The hypothesis \(\zeta>\zeta_c\) is what places
the incoming orbit on the center manifold, so that it enters the
bottleneck with velocity \(o(1)\) rather than \(O(1)\). The same
condition gives \(\alpha_{\mathrm{dyn}}=\alpha_c\), since at
\(\alpha=\alpha_c\) the orbit converges to \(\xi_f\) instead of
crossing it, while every \(\alpha>\alpha_c\) collapses by
Theorem~\ref{thm:bracket}, and Corollary~\ref{cor:mono} then excludes
all \(\alpha<\alpha_c\).

\noindent\textit{Proof of item 2.} Here \(q_f>0\) means the orbit
crosses \(\xi_f\) at
\(\alpha=\alpha_c\), so \(\alpha_{\mathrm{dyn}}<\alpha_c\) and the
equilibria at \(\alpha=\alpha_{\mathrm{dyn}}\) are nondegenerate. By
Lemma~\ref{lem:reduction} the one the orbit reaches is a hyperbolic
saddle with eigenvalues
\(\lambda_\pm=-\zeta\pm\sqrt{\zeta^{2}+g'}\). The orbit at
\(\alpha=\alpha_{\mathrm{dyn}}\) lies on the stable manifold, and its
distance \(d\) from that manifold on a fixed transversal is
\(C^{1}\) in \(\alpha\) by smooth dependence on parameters. The time
spent in a fixed neighborhood of a hyperbolic saddle at distance
\(d\) is \(-\lambda_+^{-1}\ln d+O(1)\), and the rest of the orbit
takes \(O(1)\).\hfill\(\square\)

Equation~\eqref{eq:app-log} needs in addition that \(d\) have a
nonvanishing \(\alpha\) derivative at
\(\alpha=\alpha_{\mathrm{dyn}}\), which we do not prove. Were it to
vanish to order \(k\), the coefficient would be \(-k/\lambda_+\)
instead. The measurement reported below, \(-1.361986\) against
\(-1/\lambda_+=-1.361981\), fixes \(k=1\) numerically.

Numerical check. Two constructions that involve no integration
horizon locate \(\zeta_c\), which Theorem~\ref{thm:barrier} has
already shown to be finite. Integrating \(dq/d\xi=g/q-2\zeta\) at
\(\alpha=\alpha_c\), \(\beta=0\) gives
\(q_f=1.382\times10^{-1}\), \(3.895\times10^{-2}\), and
\(2.027\times10^{-3}\) at \(\zeta=0.100\), \(0.300\), and \(0.390\),
and \(q_f=0\) for \(\zeta\ge0.396\); in the latter cases \(q\) reaches
zero along the center manifold, at \(\xi_f-\xi\approx6\times10^{-8}\)
where \(q\approx g/(2\zeta)\), and not by turning around.
Independently, continuing the strong stable manifold of the fold's
saddle-node backward to \(\xi=0\) and asking on which side of it the
release point lies gives \(\zeta_c=0.395919\) at \(\beta=0\). That
construction also supplies a geometric reading of \(\zeta_c\): it is
the damping at which the strong stable manifold of the fold passes
through the release point. The band closes quadratically, the same
backward construction returning widths at \(\beta=0\) of
\(\alpha_c-\alpha_{\mathrm{dyn}}=1.3482\times10^{-2}\),
\(1.5435\times10^{-3}\), \(3.618\times10^{-4}\), and
\(1.478\times10^{-7}\) at \(\zeta=0.1\), \(0.3\), \(0.35\), and
\(0.395\), consistent with
\(\alpha_c-\alpha_{\mathrm{dyn}}\approx0.173\,(\zeta_c-\zeta)^{2}\).
The first two agree with the bisection of
Sec.~\ref{sec:app-dyn-bracket} to four digits. The bound
\(\zeta_c\le2^{-1/4}\) is therefore sufficient and not sharp.

At \(\zeta=1.0\) and \(\zeta=0.7\) with \(\beta=0\), and at
\(\zeta=0.7\) with \(\beta=0.03\), all above \(\zeta_c\), the ratio
\(\tau_{\mathrm{PI}}\sqrt{\alpha-\alpha_c}/K\) reaches \(0.99985\),
\(0.99982\), and \(0.99945\) at \(\alpha-\alpha_c=10^{-8}\),
\(10^{-8}\), and \(10^{-7}\), and the local log-log slope reaches
\(-0.50014\).

At \(\zeta=0.1\), which lies below \(\zeta_c\), the second law
applies. A direct scan at \(\beta=0\) gives
\(\alpha_{\mathrm{dyn}}=0.13466609\) and a saddle at
\(\xi_s=0.4573962\) with \(g'(\xi_s)=0.685930\) and
\(\lambda_+=0.7342244\). The measured
\(d\tau_{\mathrm{PI}}/d\ln(\alpha-\alpha_{\mathrm{dyn}})\) is
\(-1.361986\) against \(-1/\lambda_+=-1.361981\). The archived phase
diagram of Fig.~\ref{fig:phase_diagram} carries the same signature.
Fitting the \(14\) pull-in cells nearest the boundary in each of
eight rows of \texttt{phase\_diagram.npz}, with
\(\alpha_{\mathrm{dyn}}\) free, the logarithmic model
\(\tau_{\mathrm{PI}}=A\ln(\alpha-\alpha_{\mathrm{dyn}})+B\) has the
smaller residual in every row and returns \(A\in[-1.435,-1.264]\),
while the power-law model
\(\ln\tau_{\mathrm{PI}}=p\ln(\alpha-\alpha_{\mathrm{dyn}})+c\)
returns \(p\in[-0.221,-0.193]\), far from \(-1/2\). At \(\beta=0\)
the fitted \(\alpha_{\mathrm{dyn}}=0.134667\) matches the bisected
value to \(10^{-6}\). The reported maximum
\(\tau_{\mathrm{PI}}=15.75\) corresponds, through the fitted law, to
a distance \(5.7\times10^{-6}\) from the boundary, consistent with
the \(1.005\times10^{-3}\) grid spacing shared among \(182\) boundary
rows.

The pull-in time at \(\zeta=0.1\) therefore diverges at the dynamic
boundary as \(-\ln(\alpha-\alpha_{\mathrm{dyn}})\), through escape
from a hyperbolic saddle, and not as the \(-1/2\) power of a
saddle-node bottleneck. The \(-1/2\) law of Eq.~\eqref{eq:app-K}
governs the quasi-static regime \(\zeta>\zeta_c\), in which the two
boundaries coincide. This is why the \(\zeta=0.7\) run of
Sec.~\ref{sec:results}\,B finds the dynamic boundary on the static
fold: at that damping the coincidence is exact, and the residual
\(\Delta\alpha=-6.6\times10^{-5}\) reported there is the
resolution of the sweep and not a physical gap.

\subsection{\label{sec:app-dyn-thermal}The thermal shift is one-sided}

With the correction of Eq.~\eqref{eq:lifshitz-split} the force gains
a term \(\beta\gamma/(1-\xi)^{3}\) with \(\gamma=\kappa T\ge0\), so
the equilibrium condition becomes
\(\alpha u^{-2}+\beta\gamma u^{-3}+\beta u^{-4}=1-u\).

\begin{proposition}[Monotonicity in temperature]\label{prop:thermal}
Solve the equilibrium condition for \(\alpha\),
\begin{equation}
\label{eq:app-Agamma}
\begin{split}
A_\gamma(u)&=(1-u)\,u^{2}-\frac{\beta\gamma}{u}-\frac{\beta}{u^{2}} ,\\
\alpha_c(\beta;\gamma)&=\max_{0<u<1}A_\gamma(u) .
\end{split}
\end{equation}
Then \(\partial A_\gamma/\partial\gamma=-\beta/u\), so
\(\alpha_c(\beta;\gamma)\) is strictly decreasing in \(\gamma\) for
\(\beta>0\) and independent of \(\gamma\) for \(\beta=0\), with
\begin{equation}
\label{eq:app-dalpha}
\begin{split}
\frac{\partial\alpha_c}{\partial\gamma}&=-\frac{\beta}{u_f}\le0 ,\\
\frac{\Delta V_{\mathrm{PI}}}{V_{\mathrm{PI}}}
 &=-\frac{\beta\gamma}{2\,u_f\,\alpha_c(\beta)}+O(\gamma^{2}) ,
\end{split}
\end{equation}
\(u_f\) being the maximizer. At \(\alpha=0\) the Casimir ceiling is
\(\beta^{*}(\gamma)=\max_{0<u<1}(1-u)u^{4}/(1+\gamma u)\), which is
strictly decreasing in \(\gamma\) with
\(d\beta^{*}/d\gamma=-(1-u^{*})(u^{*})^{5}/(1+\gamma u^{*})^{2}\),
equal to \(-1024/15625=-0.0655360\) at \(\gamma=0\). Heating
therefore contracts the stable region and never enlarges it, and
leaves it unchanged when \(\beta=0\).
\end{proposition}

\noindent\textit{Proof.} The equilibrium condition is linear in
\(\alpha\), which gives Eq.~\eqref{eq:app-Agamma}. Equilibria exist
if and only if \(\alpha\le\max_u A_\gamma\), and the fold is the
equality. For \(\gamma_2>\gamma_1\) and \(\beta>0\), with \(u_2\) the
maximizer at \(\gamma_2\),
\(\alpha_c(\beta;\gamma_2)=A_{\gamma_2}(u_2)<A_{\gamma_1}(u_2)
\le\alpha_c(\beta;\gamma_1)\). The derivative in
Eq.~\eqref{eq:app-dalpha} is the envelope theorem applied at the
interior maximizer, and the voltage relation follows from
\(V_{\mathrm{PI}}\propto\sqrt{\alpha_c}\). The argument for
\(\beta^{*}\) is identical, the \(\gamma\) derivative of
\((1-u)u^{4}/(1+\gamma u)\) being negative on \((0,1)\); at
\(\gamma=0\) the maximizer is \(u^{*}=4/5\), giving
\(-(1/5)(4/5)^{5}\). For \(\beta=0\), \(A_\gamma\) does not contain
\(\gamma\).\hfill\(\square\)

Solving the two fold conditions with the extra term also gives the
closed form
\(\alpha_c(u;\gamma)=u^{2}(4-5u+3\gamma u-4\gamma u^{2})/(\gamma u+2)\)
and \(\beta_c(u;\gamma)=u^{4}(3u-2)/(\gamma u+2)\), which reduce to
Eq.~\eqref{eq:fold} at \(\gamma=0\).

Numerical check. At \(d=100\,\mathrm{nm}\), where
\(\kappa=1.016\times10^{-4}\,\mathrm{K^{-1}}\), the first-order
prediction \(\beta^{*}=256/3125-0.0655360\,\gamma\) gives
\(0.081254\) at \(100\,\mathrm{K}\) and \(0.079922\) at
\(300\,\mathrm{K}\), against \(0.081260\) and \(0.079971\) from the
exact maximization and \(0.0813\) and \(0.0799\) as reported in
Sec.~\ref{sec:results}\,C. At \(\beta=0.03\), where
\(u_f=0.735143\) and \(\alpha_c=0.0876272\),
Eq.~\eqref{eq:app-dalpha} gives
\(\Delta V_{\mathrm{PI}}/V_{\mathrm{PI}}=-0.237\%\) and
\(-0.710\%\), against \(-0.2368\%\) and \(-0.7120\%\) from the exact
fold and the \(-0.24\%\) and \(-0.71\%\) quoted in the main text.

\subsection{\label{sec:app-dyn-gronwall}An a posteriori bound from the
collocation residual}

The bound below certifies the network error from its own residual,
without reference to a converged solution. Write
\(z=(\theta,\dot\theta)\), let
\(F(z)=(\dot\theta,f(\theta,\dot\theta))\) with \(f\) the right-hand
side of Eq.~\eqref{eq:thetaode}, and let \(z_{\rm NN}\) be the network
pair, which satisfies \(\dot z_{\rm NN}=F(z_{\rm NN})+(0,r_\theta)\)
with \(r_\theta\) the residual of Eq.~\eqref{eq:residual-theta}.

\begin{theorem}[A posteriori error bound]\label{thm:gronwall}
Let \(T<\tau_*\), and suppose \(\theta\) and \(\theta_{\rm NN}\) both
take values in \([0,\Theta]\) and their derivatives in
\([-\Omega,\Omega]\) on \([0,T]\). Set
\begin{equation}
\label{eq:app-L}
L=\max\Bigl\{1,\;
 e^{\Theta}+3\alpha e^{3\Theta}+5\beta e^{5\Theta}
 +2(\Omega+\zeta)\Bigr\} .
\end{equation}
Then \(L\) is a Lipschitz constant for \(F\) in the norm
\(\|z\|=\max(|\theta|,|\dot\theta|)\), and if the initial conditions
are exact, as they are under the hard ansatz of
Eq.~\eqref{eq:hardic-theta},
\begin{equation}
\label{eq:app-gronwall}
\bigl\|\theta_{\rm NN}-\theta\bigr\|_{\infty,[0,T]}
\;\le\;\frac{\|r_\theta\|_{\infty,[0,T]}}{L}\,
 \bigl(e^{LT}-1\bigr) ,
\end{equation}
together with the following local form. Fix \(R>0\), let
\(L_R(\sigma)\) be Eq.~\eqref{eq:app-L} evaluated with \(\Theta\) and
\(\Omega\) replaced by \(\theta(\sigma)+R\) and
\(|\dot\theta(\sigma)|+R\), and put
\begin{equation}
\label{eq:app-gronwall-local}
B_R(\tau)=\|r_\theta\|_{\infty}\int_0^{\tau}
 \exp\!\left[\int_s^{\tau}L_R(\sigma)\,d\sigma\right]ds .
\end{equation}
If \(B_R\le R\) on \([0,\tau]\), then
\(|\theta_{\rm NN}-\theta|\le B_R\) there. Because
\(|d\xi/d\theta|=e^{-\theta}\le1\) for \(\theta\ge0\), the same
numbers bound the error in \(\xi\).
\end{theorem}

\noindent\textit{Proof.} The Jacobian of \(F\) has rows \((0,1)\) and
\((\partial_\theta f,\partial_{\dot\theta}f)\), with
\(\partial_\theta f=-e^{\theta}+3\alpha e^{3\theta}
+5\beta e^{5\theta}\) and
\(\partial_{\dot\theta}f=2\dot\theta-2\zeta\). The operator norm
induced by \(\|\cdot\|\) is the largest row sum of absolute values,
which on the stated set is bounded by Eq.~\eqref{eq:app-L}; the set
is convex, so \(L\) is a Lipschitz constant. The error
\(e=z_{\rm NN}-z\) obeys \(\dot e=F(z_{\rm NN})-F(z)+(0,r_\theta)\),
hence \(D^{+}\|e\|\le L\|e\|+|r_\theta|\) with \(\|e(0)\|=0\), and the
integral form of Gronwall's inequality gives
Eq.~\eqref{eq:app-gronwall}. For the local form, assume \(B_R\le R\)
on \([0,\tau]\) and let \(\tau_R\) be the supremum of the times
\(t\le\tau\) with \(\|e\|\le R\) on \([0,t]\), which is positive
because \(e(0)=0\). On \([0,\tau_R]\) the segment joining
\(z(\sigma)\) to \(z_{\rm NN}(\sigma)\) lies in the tube of radius
\(R\) about the reference orbit, where \(L_R(\sigma)\) bounds the
Jacobian, so \(D^{+}\|e\|\le L_R\|e\|+|r_\theta|\), and Gronwall
gives \(\|e\|\le B_R\) there. If \(\tau_R<\tau\) then
\(\|e(\tau_R)\|=R\le B_R(\tau_R)\), while \(B_R\) is strictly
increasing with \(B_R(\tau)\le R\), which forces \(\tau_R=\tau\).
Hence \(\|e\|\le B_R\) on \([0,\tau]\).
Finally \(|\xi_{\rm NN}-\xi|=|e^{-\theta_{\rm NN}}-e^{-\theta}|
\le|\theta_{\rm NN}-\theta|\) by the mean value
theorem.\hfill\(\square\)

Two caveats precede the numbers. The logged training loss is a mean
square over collocation points, so
\(\sqrt{\mathcal L_{\rm ODE}^{\theta}}\) is a root-mean-square and not
a supremum; and a residual sampled at \(N\) points does not majorize
the continuous residual, which would require either a quadrature
bound or a bound on \(\|\dot r_\theta\|_\infty\) together with the
fill distance. The evaluation below substitutes
\(\|r_\theta\|_\infty\approx\sqrt{\mathcal L_{\rm ODE}^{\theta}}\) and
is a certificate only under that substitution.

Numbers for the pull-in run of Fig.~\ref{fig:trajectories}(c)
(\(\alpha=0.2\), \(\beta=0.1\), \(\zeta=0.5\), \(T=2.5698\)): the
final logged loss is
\(\mathcal L_{\rm ODE}^{\theta}=1.1935\times10^{-6}\), so
\(\|r_\theta\|\approx1.09\times10^{-3}\). On the window,
\(\Theta=1.3562\) and \(\Omega=7.3972\), giving \(L=495.2\), of which
\(440.4\) comes from the \(5\beta e^{5\Theta}\) term alone. Then
\(LT=1272\) and Eq.~\eqref{eq:app-gronwall} is vacuous: the bound
overflows. The local form Eq.~\eqref{eq:app-gronwall-local}, taken
with \(R\) equal to the target error so that the hypothesis
\(B_R\le R\) is self-consistent, stays finite but still runs far
above the measured error. It certifies
\(|\theta_{\rm NN}-\theta|\le2.83\times10^{-3}\), the measured
sup-norm displacement error, up to \(\tau=0.663\)
(\(\xi=0.053\)), and \(\le10^{-2}\) up to \(\tau=0.953\)
(\(\xi=0.101\)) and \(\le10^{-1}\) up to \(\tau=1.240\)
(\(\xi=0.158\)). Widening the tube from the reference orbit to
radius \(R\) costs little: a bound evaluated on the reference orbit
alone, which is not a certificate, would reach \(0.666\),
\(0.968\), and \(1.442\).

That gap is the expected behavior of a Gronwall estimate, which
propagates a worst-case Lipschitz constant rather than the actual
contraction along the flow. Its value here is qualitative, the
scaling that follows being an order estimate and not a bound: the
certificate degrades as \(e^{5\theta}\), so the certifiable horizon
shrinks like \(\theta_{\max}^{-1}\) as the pole is approached, while
by Theorem~\ref{thm:exponent} \(\theta_{\max}\) grows only like
\(-\tfrac25\ln(\tau_*-T)\). A quantitative certificate over the whole
window needs a bound that follows the flow, such as a logarithmic
norm or a validated integrator, which lies outside the scope of this
work.

%% file: sections/sm_body.tex
%

\section{\label{sec:S-training}Training protocol for the trajectory
networks}

This Supplemental Material collects the settings, tolerances, and validation
numbers behind the figures of the main text. Abbreviations are respelled here
because the two files are read separately: PINN for physics-informed neural
network, ODE for ordinary differential equation, NEMS for
nanoelectromechanical system, RK4 for the fixed-step fourth-order Runge--Kutta
scheme, DOP853 for the adaptive eighth-order Dormand--Prince scheme, RMSE for
root-mean-square error, and L-BFGS for the limited-memory
Broyden--Fletcher--Goldfarb--Shanno optimizer.

Three networks represent the three dynamical regimes of the dimensionless
equation of motion. Each maps the dimensionless time $\tau$ to the normalized
displacement $\xi$, or to the rapidity $\theta=-\log(1-\xi)$ in the pull-in
case. Table~\ref{tab:S1} lists every setting that was varied between them.

Two entries carry the design decisions. The stable trajectory runs to
$T=60$, close to eight damped periods, where a plain $\tanh$ network floors
near $E_\infty\approx6.5\times10^{-3}$. Two Fourier input modes at the damped
frequency $\omega_d=\sqrt{1-\zeta^{2}}\approx0.8$, and not at the undamped
$\omega=1$, together with $1800$ L-BFGS steps, bring the error to
$1.05\times10^{-3}$. The pull-in trajectory needs no Fourier
features: its window ends at $\tau_*\approx2.62$, before a single period
completes, and the difficulty there is the movable pole rather than
oscillation.

Under the rapidity parametrization together with the hard initial-condition
ansatz, both the initial-condition loss and the soft pull-in barrier vanish
identically, so their weights $\lambda_1$ and $\lambda_2$ are inert and are
listed only for completeness. For the raw parametrization, where the barrier
is active, $N\sim512$--$1024$ collocation points and $\lambda_2\sim10$ are
adequate starting values.

\input{sections/sm_tab_s1_training}

\section{\label{sec:S-foldpinn}Parametric surrogate for the fold locus}

The second network takes the control pair $(\alpha,\beta)$ as input and returns
the equilibrium gap, trained on the force-balance residual alone with no
labeled solution data. Table~\ref{tab:S2} gives its architecture and its
accuracy against the closed-form fold.

Two accuracy figures describe the same network, and quoting either alone
misrepresents it. The mean deviation of the learned fold from the analytic
locus is $7.9\times10^{-8}$ in $\alpha$. That number flatters the surrogate:
the fold is a double-root locus, where the pull-in indicator is second-order
insensitive to the underlying field, so a field error of order $\epsilon$
displaces the boundary by order $\epsilon^{2}$. The fair measure is the error
of the equilibrium field itself, $\mathrm{RMSE}(u_*)=3.1\times10^{-3}$, with a
pointwise maximum of $0.126$. That maximum is confined to the corner
$\alpha,\beta\to0$, where the equilibrium sits at $\xi\to0$ and the minimizer
that defines $u_*$ is ill-conditioned, since the force balance there is flat in
$u$ over a wide interval. Away from that corner the field error stays below
$10^{-2}$.

Cell-by-cell classification separates the two error scales. Against the static
fold the surrogate labels $100.00\%$ of grid cells correctly. Against the
from-rest RK4 outcome it reaches $94.70\%$, which is exactly the score of the
static fold itself on the same test: the residual $5.30\%$ is the
kinetic-overshoot band of Sec.~\ref{sec:S-phase}, not a defect of the network:
a surrogate fitted to the static problem cannot see a dynamic effect.

\input{sections/sm_tab_s2_foldpinn}

\section{\label{sec:S-phase}Numerical construction of the phase diagram}

The phase diagram integrates the equation of motion from rest on a
$200\times200$ grid, once at $\zeta=0.1$ and once at $\zeta=0.7$.
Table~\ref{tab:S3} records the grid, the integrator settings, and the summary
statistics quoted in the main text.

Three choices set the reported numbers. The pull-in event fires at
$\xi\ge1-\epsilon$ with $\epsilon=10^{-6}$, which corresponds to a physical gap
of $0.1\,\mathrm{pm}$ at $d=100\,\mathrm{nm}$ and is far inside any regime the
continuum force model describes; the event time $\tau_{\mathrm{PI}}$ is
insensitive to $\epsilon$ because the last decade of gap closure occupies a few
parts in $10^{3}$ of the trajectory. The horizon $\tau_{\max}=200$ is $20$
damping times at $\zeta=0.1$, long enough that a cell not having collapsed by
then has settled. A guard flags any trajectory reaching $\xi<-10^{-9}$; no cell
triggered it in either sweep.

The dynamic boundary lies inside the static fold by
$\Delta\alpha$, with mean $0.0117$, median $0.0121$, and maximum $0.0141$ over
$182$ boundary samples at $\zeta=0.1$. Raising the damping to $\zeta=0.7$ drives
the mean gap to $-6.6\times10^{-5}$, one fifteenth of a grid cell in $\alpha$ and
of the opposite sign, so the two boundaries coincide to grid resolution and the
quasi-static picture is recovered. The overshoot band covers $5.30\%$ of the
sampled plane at $\zeta=0.1$, matching the classification deficit of
Sec.~\ref{sec:S-foldpinn} to the last digit, as it must, since both count the
same cells.

A cross-check compares the pull-in flag and the event time from the
phase-diagram sweep against an independent integration of the same cells. The
two agree on
every cell, with a largest event-time difference of $8.6\times10^{-14}$ and no
difference in $\xi$ at all.

\input{sections/sm_tab_s3_phase}

\section{\label{sec:S-rk4}Fixed-step Runge--Kutta at the movable pole}

The claim in the main text is representational rather than a statement about
accuracy: away from collapse RK4 converges at its nominal order, and at the
pole it steps into a region the model does not define. Table~\ref{tab:S4}
supports both halves.

Over $[0,0.98\tau_*]$ the sup-norm error falls from $7.6\times10^{-5}$ at
$h=0.05$ to $5.3\times10^{-11}$ at $h=0.001$. For $h\le0.01$ the ratio between
successive refinements reproduces the fourth-order prediction, $16.0$ against
$15.6$, $39.1$ against $38.8$, and $16.0$ against $16.1$, so the scheme attains
its nominal order on the regular part of the trajectory. The two coarsest steps
fall short of it, reaching $47\%$ and $51\%$ of the predicted ratio. Near the
pole the behavior does not improve with $h$. All six step sizes carry the state
past
$\xi=1$, where $(1-\xi)^{-4}$ changes sign and the attractive load turns into a
repulsive one. The largest finite value reached afterwards is erratic, ranging
from $46$ to $4057$ without any trend in $h$, which is the signature of a
divergent map rather than of a discretization error. The estimated collapse
time is not monotone in $h$ either.

An error-controlled solver behaves differently on the same problem. DOP853 at
$\mathrm{rtol}=\mathrm{atol}=10^{-12}$ refines its step where the field
steepens and halts at $\xi=0.999996$, a minimum gap of $3.9\times10^{-6}$,
with $\tau_*=2.6225672$, no overshoot, and no non-finite value. The rapidity
network keeps $\xi<1$ by construction, since $\xi=1-e^{-\theta}$ with $\theta$
finite.

\input{sections/sm_tab_s4_rk4}

\section{\label{sec:S-sens}Design sensitivities by automatic differentiation}

Differentiating the trained fold surrogate with respect to its inputs returns
design sensitivities without a finite-difference stencil.
Table~\ref{tab:S5} compares three of them against a central difference of the
closed-form pipeline at devices~A and~B, which bracket the Casimir loading of
the tabulated devices.

The three elasticities $\partial\ln V_{\mathrm{PI}}/\partial\ln p$ read
$+2.789$ for the gap, $-0.758$ for the area, and $+0.758$ for the stiffness.
The last two are equal and opposite because $V_{\mathrm{PI}}$ depends on $A$
and $k$ only through the ratio $A/k$ at fixed $\beta$, so the surrogate
reproduces a symmetry of the problem it was never told about. The gap
elasticity is not the naive $3/2$ of the purely electrostatic law
$V_{\mathrm{PI}}\propto d^{3/2}$: the Casimir loading makes $\alpha_c$ itself a
function of $d$ through $\beta\propto d^{-5}$, and the chain rule adds the
difference. That correction grows with the loading: at device~B, where
$\beta=0.052$, the gap elasticity is $+6.042$ against $+2.789$ at device~A,
while the area and stiffness elasticities move to $\mp1.408$ and keep their
equal-and-opposite relation. The quantity being differentiated is therefore
far from constant over the device range, and the surrogate follows it: the
relative error against the closed form is $3.4\times10^{-6}$ or better at
device~A and $1.4\times10^{-6}$ or better at device~B.

The boundary slope $d\alpha_c/d\beta$ is available in closed form as $-1/u^{2}$
and provides a stricter test, since it is a derivative of the fold rather than
of a quantity built from it. Differentiating through the network gives
$-2.0691$, $-1.8854$, and $-1.7117$ at $\beta=0.01$, $0.026$, and $0.05$,
against analytic values that differ in the sixth digit; the largest relative
error over the three points is $7.3\times10^{-6}$.

\input{sections/sm_tab_s5_sensitivity}

\section{\label{sec:S-inverse}Gradient-based inverse design}

The same surrogate inverts a specification. Fixing $A=100\,\mu\mathrm{m}^{2}$
and $k=0.5\,\mathrm{N/m}$ and targeting $V_{\mathrm{PI}}=0.30\,\mathrm{V}$,
Adam descends on the gap through the differentiable boundary.
Table~\ref{tab:S6} lists the outcome.

The optimizer reaches its tolerance after $201$ steps and
$1.42\,\mathrm{s}$, recovering $d=97.0364\,\mathrm{nm}$ against the analytic
root $97.0364\,\mathrm{nm}$, an absolute error of
$1.7\times10^{-5}\,\mathrm{nm}$.
The realized pull-in voltage lands within $8.8\times10^{-9}\,\mathrm{V}$ of the
target. Applied instead to the
from-rest dynamic boundary $\alpha_{\mathrm{dyn}}(\beta)$, which is defined
only implicitly by the integrated equation of motion, a surrogate fitted to
$33$ RK4-located boundary points with a fit RMSE of $1.5\times10^{-4}$ in
$\alpha$ returns $d\alpha_{\mathrm{dyn}}/d\beta$ as $-2.06$, $-1.80$, and
$-1.64$ at $\beta=0.01$, $0.03$, and $0.05$. Each matches a central finite
difference of the independently recomputed boundary to $0.09\%$, $0.01\%$, and
$0.5006\%$ respectively, at a finite-difference step $h=2.5\times10^{-3}$ and a
bisection tolerance of $10^{-6}$ in $\alpha$.

\input{sections/sm_tab_s6_inverse}

\section{\label{sec:S-devices}Representative device parameters}

Table~\ref{tab:S7} gives the full parameter set of the four devices summarized
in the main text, including the columns omitted there. The geometries are
order-of-magnitude engineering values spanning published Casimir MEMS and NEMS
oscillators rather than the parameters of any one fabricated device; no
quantitative benchmark against a measured pull-in voltage is claimed.

The effective mass follows from the quoted stiffness and resonance,
$m=k/(2\pi f_0)^{2}$, and the damping ratio from the quality factor,
$\zeta=1/(2Q)$. All four devices lie below the Casimir ceiling
$\beta^{*}=256/3125$, so each has a finite pull-in voltage. Device~C, with a
$200\,\mathrm{nm}$ gap and a stiffness of $5\,\mathrm{N/m}$, sits at
$\beta=8.1\times10^{-5}$ and is electrostatically dominated: its
$\xi_{\mathrm{PI}}=0.3331$ differs from the classical $1/3$ in the fourth
decimal. Device~B, at a $50\,\mathrm{nm}$ gap, carries $\beta=0.052$, which is
$63\%$ of the ceiling, and its pull-in displacement has moved to $0.2331$.

\input{sections/sm_tab_s7_devices}

\section{\label{sec:S-bandscan}Closure of the kinetic-overshoot band}

The band between the from-rest threshold and the static fold shrinks with
damping and closes at a finite damping ratio. Table~\ref{tab:S8} resolves that
closure at $\beta=0$ on eleven values of $\zeta$, so the approach can be read
rather than inferred from the four values quoted in the appendix.

The widths are located by bisection on the crossing of the saddle, a
finite-time decidable event, and the classification is safe because the node
below it is hyperbolic and is reached on a damping time. That criterion is
unrelated to the backward continuation of the strong stable manifold used in
the appendix, and the two agree to $6\times10^{-6}$ at the two damping ratios
where both were run. The scan stops at $\zeta=0.35$, where the width is still
$3.6\times10^{-4}$: closer to $\zeta_c$ the width falls below what a finite
integration horizon resolves, which is the one weakness of this criterion.

The closure is quadratic. Fitting $\alpha_c-\alpha_{\mathrm{dyn}}=
C(\zeta_c-\zeta)^{2}$ over $\zeta\ge0.20$ with $\zeta_c=0.3959$ gives
$C=0.167$, against $C=0.173$ from the four appendix values over a wider range;
the two differ by $4\%$, the spread expected from fitting a leading-order law
on different intervals.

\input{sections/sm_tab_s8_bandscan}

\section{\label{sec:S-archive}Archive layout and reproducibility}

The analysis code, the trained weights, the training logs, and the scripts that
regenerate every figure and every number quoted above are archived at
Zenodo~\cite{Zenodo2026} under a fixed seed of $42$, CPU-only, in double
precision. A full rerun through \texttt{run\_all.py} reproduces the archived
outputs bit-for-bit on the same interpreter.

The mapping from archived file to published quantity is direct.
\texttt{logs/history\_\{stable,growing,pullin\}.json} hold the configurations
and loss histories of Table~\ref{tab:S1}; \texttt{logs/pinn\_boundary.json}
holds Table~\ref{tab:S2}; \texttt{logs/phase\_diagram.json} holds
Table~\ref{tab:S3} and the data behind Fig.~2 of the main text;
\texttt{results/rk4\_stepstudy.csv} and
\texttt{results/verification\_summary.json} hold Table~\ref{tab:S4} and the
per-regime errors of the main text; \texttt{results/inverse\_design.json} holds
Tables~\ref{tab:S5} and~\ref{tab:S6} and Fig.~5;
\texttt{results/dynamic\_boundary\_design.json} holds Fig.~6;
\texttt{results/lifshitz\_boundaries.csv} holds Fig.~3; and
\texttt{results/device\_table.csv} holds Table~\ref{tab:S7} and Table~I of the
main text. The LaTeX fragments of every table above are emitted from those
files by \texttt{src/make\_sm\_tables.py}, so no number in this Supplemental
Material was transcribed by hand.

%% file: sections/sm_tab_s1_training.tex
\begin{table}[htb]
\caption{\label{tab:S1}%
Training protocol of the three single-trajectory physics-informed
neural networks. All runs are CPU-only, in double precision, and
seeded with $42$. Wall-clock times were measured on one core.}
\begin{ruledtabular}
\begin{tabular}{lccc}
Setting & Stable & Growing & Pull-in\\
\colrule
Parametrization & raw & raw & rapidity\\
$\alpha$ & $0.05$ & $0.0724219$ & $0.2$\\
$\beta$ & $1.000\times10^{-4}$ & $0.0407373$ & $0.1$\\
$\zeta$ & $0.6$ & $1$ & $0.5$\\
Training window $T$ & $60$ & $30$ & $2.57$\\
Hidden layers & $4$ & $4$ & $4$\\
Width & $64$ & $64$ & $64$\\
Fourier modes & $2$ & $0$ & $0$\\
Fourier $\omega$ & $0.8$ & $1$ & $1$\\
Collocation $N$ & $3000$ & $2000$ & $1600$\\
Adam steps & $1500$ & $2000$ & $3000$\\
Adam rate & $2.0\times10^{-3}$ & $1.0\times10^{-3}$ & $1.0\times10^{-3}$\\
L-BFGS steps & $1800$ & $500$ & $500$\\
Final loss & $2.42\times10^{-7}$ & $1.04\times10^{-6}$ & $1.19\times10^{-6}$\\
Wall clock (s) & $236$ & $108$ & $124$\\
\end{tabular}
\end{ruledtabular}
\end{table}

%% file: sections/sm_tab_s2_foldpinn.tex
\begin{table}[htb]
\caption{\label{tab:S2}%
Parametric equilibrium surrogate: architecture, optimization, and
accuracy against the closed-form fold of the main text. The last two
rows classify each cell of an $(\alpha,\beta)$ grid as pull-in or
stable.}
\begin{ruledtabular}
\begin{tabular}{lc}
Quantity & Value\\
\colrule
Input dimension & $2$\\
Hidden layers & $3$\\
Width & $64$\\
Collocation points & $8100$\\
Adam steps & $4000$\\
L-BFGS steps & $300$\\
Final residual loss & $1.69\times10^{-7}$\\
Wall clock (s) & $78$\\
\colrule
RMSE of the pull-in gap $u_*$ & $3.10\times10^{-3}$\\
Maximum $|\Delta u_*|$ & $1.26\times10^{-1}$\\
Mean $|\Delta\alpha|$ on the fold & $7.95\times10^{-8}$\\
Maximum $|\Delta\alpha|$ on the fold & $3.45\times10^{-7}$\\
$\alpha_c$ at $\beta=0$, surrogate & $0.1481478$\\
$\alpha_c$ at $\beta=0$, exact $4/27$ & $0.1481481$\\
$\beta^{*}$ at $\alpha=0$, surrogate & $0.0819200$\\
$\beta^{*}$ at $\alpha=0$, exact $256/3125$ & $0.0819200$\\
Cells matching the static fold (\%) & $100.00$\\
Cells matching from-rest RK4 (\%) & $94.70$\\
\end{tabular}
\end{ruledtabular}
\end{table}

%% file: sections/sm_tab_s3_phase.tex
\begin{table}[htb]
\caption{\label{tab:S3}%
Construction of the pull-in phase diagram of the main text. Each grid
cell is integrated from rest until the gap closes or the horizon is
reached. The gap $\Delta\alpha$ is the horizontal offset between the
dynamic boundary and the static fold.}
\begin{ruledtabular}
\begin{tabular}{lc}
Quantity & Value\\
\colrule
Grid in $\alpha$ & $200$ points on $[0,0.2]$\\
Grid in $\beta$ & $200$ points on $[0,0.09]$\\
Integrator & fixed-step RK4\\
Step $\Delta\tau$ & $0.005$\\
Horizon $\tau_{\max}$ & $200$\\
Pull-in criterion & $\xi\ge1-\epsilon$, $\epsilon=1\times10^{-6}$\\
Escape guard & $\xi<-1\times10^{-9}$\\
Cells that escaped & $0$\\
\colrule
Damping $\zeta$ & $0.1$\\
Dynamic pull-in fraction & $0.7343$\\
Static pull-in fraction & $0.6813$\\
Overshoot band, fraction of cells & $0.0530$\\
Mean gap $\Delta\alpha$ & $0.01171$\\
Median gap $\Delta\alpha$ & $0.01206$\\
Maximum gap $\Delta\alpha$ & $0.01407$\\
Boundary samples & $182$\\
Range of $\tau_{\mathrm{PI}}$ & $2.17$--$15.75$\\
\colrule
Damping $\zeta$ & $0.7$\\
Dynamic pull-in fraction & $0.6810$\\
Mean gap $\Delta\alpha$ & $-6.63\times10^{-5}$\\
\colrule
Fold endpoint $\alpha_c(\beta=0)$ & $0.148148$\\
Fold endpoint $\beta^{*}(\alpha=0)$ & $0.081920$\\
Integration time (s) & $38$ and $48$\\
\end{tabular}
\end{ruledtabular}
\end{table}

%% file: sections/sm_tab_s4_rk4.tex
\begin{table}[htb]
\caption{\label{tab:S4}%
Fixed-step fourth-order Runge--Kutta at the movable pole, for the
pull-in case $\alpha=0.2$, $\beta=0.1$, $\zeta=0.5$. The error column
is the sup-norm in $\xi$ over $[0,0.98\tau_*]$, away from the pole;
$\xi_{\max}$ is the largest finite value reached after the scheme
steps past $\xi=1$. Every step size overshoots.}
\begin{ruledtabular}
\begin{tabular}{cccccc}
$h$ & Steps & $E_\infty(\xi)$ & $\tau_{\mathrm{PI}}$ & $|\Delta\tau_{\mathrm{PI}}|$ & $\xi_{\max}$\\
\colrule
$0.05$ & 84 & $7.60\times10^{-5}$ & $2.6012$ & $2.14\times10^{-2}$ & $166$\\
$0.02$ & 210 & $4.18\times10^{-6}$ & $2.6216$ & $9.33\times10^{-4}$ & $46$\\
$0.01$ & 420 & $5.12\times10^{-7}$ & $2.6200$ & $2.55\times10^{-3}$ & $4057$\\
$0.005$ & 840 & $3.28\times10^{-8}$ & $2.6206$ & $2.00\times10^{-3}$ & $142$\\
$0.002$ & 2099 & $8.46\times10^{-10}$ & $2.6220$ & $5.49\times10^{-4}$ & $2530$\\
$0.001$ & 4197 & $5.27\times10^{-11}$ & $2.6222$ & $3.30\times10^{-4}$ & $177$\\
\end{tabular}
\end{ruledtabular}
\end{table}

%% file: sections/sm_tab_s5_sensitivity.tex
\begin{table*}[htb]
\caption{\label{tab:S5}%
Design sensitivities of the pull-in voltage taken by automatic
differentiation through the fold surrogate, against a finite-difference
evaluation of the closed-form fold, at the two devices that bracket the
Casimir loading of Table~I. Device~A carries $\beta=0.026$ and
device~B $\beta=0.052$. The elasticity is
$\partial\ln V_{\mathrm{PI}}/\partial\ln p$.}
\begin{ruledtabular}
\begin{tabular}{llcccc}
Device & Derivative & Surrogate & Reference & Rel. error & Elasticity\\
\colrule
A & $\partial V_{\mathrm{PI}}/\partial d$ (V/m) & $9.1396\times10^{6}$ & $9.1396\times10^{6}$ & $3.4\times10^{-6}$ & $2.7889$\\
 & $\partial V_{\mathrm{PI}}/\partial A$ (V/m$^{2}$) & $-2.4834\times10^{9}$ & $-2.4833\times10^{9}$ & $2.4\times10^{-6}$ & $-0.7578$\\
 & $\partial V_{\mathrm{PI}}/\partial k$ (Vm/N) & $4.9667\times10^{-1}$ & $4.9667\times10^{-1}$ & $2.4\times10^{-6}$ & $0.7578$\\
\colrule
B & $\partial V_{\mathrm{PI}}/\partial d$ (V/m) & $4.0067\times10^{7}$ & $4.0066\times10^{7}$ & $1.4\times10^{-6}$ & $6.0419$\\
 & $\partial V_{\mathrm{PI}}/\partial A$ (V/m$^{2}$) & $-1.8679\times10^{10}$ & $-1.8679\times10^{10}$ & $1.2\times10^{-6}$ & $-1.4084$\\
 & $\partial V_{\mathrm{PI}}/\partial k$ (Vm/N) & $2.3349\times10^{-1}$ & $2.3349\times10^{-1}$ & $1.2\times10^{-6}$ & $1.4084$\\
\end{tabular}
\end{ruledtabular}
\end{table*}

%% file: sections/sm_tab_s6_inverse.tex
\begin{table}[htb]
\caption{\label{tab:S6}%
Gradient-based inverse design through the same surrogate: the smallest
gap that still holds a target pull-in voltage of
$0.30\,\mathrm{V}$ at fixed area and stiffness.}
\begin{ruledtabular}
\begin{tabular}{lc}
Quantity & Value\\
\colrule
Recovered gap $d$ (nm) & $97.03641$\\
Closed-form gap $d$ (nm) & $97.03639$\\
Absolute error in $d$ (nm) & $1.70\times10^{-5}$\\
Recovered $\beta$ & $0.030223$\\
Absolute error in $V_{\mathrm{PI}}$ (V) & $8.80\times10^{-9}$\\
Adam steps to tolerance & $201$\\
Wall clock (ms) & $6303$\\
Final objective & $7.74\times10^{-17}$\\
\end{tabular}
\end{ruledtabular}
\end{table}

%% file: sections/sm_tab_s7_devices.tex
\begin{table*}[htb]
\caption{\label{tab:S7}%
Parameters of the representative devices. The Casimir strength is
$\beta=\pi^{2}\hbar cA/(240kd^{5})$ and the effective mass follows
from $k$ and $f_0$. The last column is the classical-limit thermal
bound on the pull-in voltage at $300\,\mathrm{K}$ derived in the
Appendix.}
\begin{ruledtabular}
\begin{tabular}{llcccccccc}
Device & Description & $A$ ($\mu\mathrm{m}^{2}$) & $d$ (nm) & $k$ (N/m) & $Q$ & $\beta$ & $\xi_{\mathrm{PI}}$ & $V_{\mathrm{PI}}$ (V) & $\Delta V_{\mathrm{PI}}$ (\%)\\
\colrule
A & MEMS gold plate & $100$ & $100$ & $0.5$ & $10000$ & $0.02600$ & $0.2717$ & $0.3277$ & $-0.57$\\
B & sub-100 nm NEMS & $25$ & $50$ & $2$ & $1000$ & $0.05201$ & $0.2331$ & $0.3316$ & $-1.07$\\
C & stiff NEMS & $100$ & $200$ & $5$ & $5000$ & $0.00008$ & $0.3331$ & $3.6564$ & $-0.00$\\
D & intermediate-gap NEMS & $50$ & $75$ & $1$ & $2000$ & $0.02739$ & $0.2693$ & $0.4198$ & $-0.46$\\
\end{tabular}
\end{ruledtabular}
\end{table*}

%% file: sections/sm_tab_s8_bandscan.tex
\begin{table}[htb]
\caption{\label{tab:S8}%
Closure of the kinetic-overshoot band with damping, at $\beta=0$.
The width is $\alpha_c-\alpha_{\mathrm{dyn}}$, located by bisection on
the crossing of the saddle. The last column is the quadratic law
$0.167\,(\zeta_c-\zeta)^{2}$ with $\zeta_c=0.3959$, fitted over
$\zeta\ge0.20$.}
\begin{ruledtabular}
\begin{tabular}{ccc}
$\zeta$ & $\alpha_c-\alpha_{\mathrm{dyn}}$ & quadratic law\\
\colrule
$0.05$ & $1.804\times10^{-2}$ & $1.992\times10^{-2}$\\
$0.08$ & $1.524\times10^{-2}$ & $1.662\times10^{-2}$\\
$0.10$ & $1.348\times10^{-2}$ & $1.458\times10^{-2}$\\
$0.13$ & $1.102\times10^{-2}$ & $1.177\times10^{-2}$\\
$0.16$ & $8.788\times10^{-3}$ & $9.266\times10^{-3}$\\
$0.20$ & $6.165\times10^{-3}$ & $6.390\times10^{-3}$\\
$0.24$ & $3.972\times10^{-3}$ & $4.047\times10^{-3}$\\
$0.28$ & $2.234\times10^{-3}$ & $2.237\times10^{-3}$\\
$0.30$ & $1.544\times10^{-3}$ & $1.532\times10^{-3}$\\
$0.32$ & $9.756\times10^{-4}$ & $9.596\times10^{-4}$\\
$0.35$ & $3.618\times10^{-4}$ & $3.510\times10^{-4}$\\
\end{tabular}
\end{ruledtabular}
\end{table}